\documentclass[11p,reqno]{amsart}

\usepackage[T1]{fontenc}
\usepackage{graphicx}
\usepackage{amssymb,amsthm,amsmath,mathrsfs,bm,braket,marginnote}
\usepackage{enumerate}
\usepackage{appendix}
\usepackage[colorlinks=true, pdfstartview=FitV, linkcolor=blue, citecolor=blue, urlcolor=blue]{hyperref}
\usepackage{multirow}

\usepackage{pgf}
\usepackage{pgfplots}
\usepackage{tikz}
\usetikzlibrary{arrows,calc}
\usepackage{verbatim}
\usetikzlibrary{decorations.pathreplacing,decorations.pathmorphing}
\usepackage[numbers,sort&compress]{natbib}
\usepackage{dsfont}

\numberwithin{equation}{section}
\newtheorem{theorem}{Theorem}[section]

\newtheorem{definition}[theorem]{Definition}
\newtheorem{coro}[theorem]{Corollary}

\newtheorem{proposition}[theorem]{Proposition}

\theoremstyle{remark}
\newtheorem{remark}[theorem]{Remark}

 \reversemarginpar
 \newcommand{\pf}{\text{pf}}
 \newcommand{\Pf}{\text{Pf}}

\newcommand{\qf }{{q^{\frac{1}{2}}}}
\newcommand{\al}{{\alpha}}
\newcommand{\be}{{\beta}}

\newcommand{\mb}{\mathcal{B}}

 \newcommand{\tw}{\tilde{\omega}}

\begin{document}

\title{Discrete Orthogonal Ensemble on the Exponential Lattices}

\date{}

\author{Shi-Hao Li}
\address{Department of Mathematics, Sichuan University, Chengdu, 610064, China}
\email{lishihao@lsec.cc.ac.cn}

\author{Bo-Jian Shen}
\address{School of Mathematical Sciences, Shanghai Jiaotong University, People's Republic of China.}
\email{JOHN-EINSTEIN@sjtu.edu.cn}

\author{Guo-Fu Yu}
\address{School of Mathematical Sciences, Shanghai Jiaotong University, People's Republic of China.}
\email{gfyu@sjtu.edu.cn}

\author{Peter J. Forrester}
\address{School of Mathematical and Statistics, ARC Centre of Excellence for Mathematical and Statistical Frontiers, The University of Melbourne, Victoria 3010, Australia}
\email{pjforr@unimelb.edu.au}

\begin{abstract}
Inspired by Aomoto's $q$-Selberg integral, the orthogonal ensemble in the exponential lattice is considered in this paper. By introducing a skew symmetric kernel, the configuration space of this ensemble is constructed to be symmetric and thus, corresponding skew inner product, skew orthogonal polynomials as well as correlation functions are explicitly formulated. Examples including Al-Salam \& Carlitz, $q$-Laguerre, little $q$-Jacobi and big $q$-Jacobi cases are considered.
\end{abstract}

\dedicatory{}


\maketitle

\section{Introduction}
\subsection{Background}
Random matrix theory has many connections within mathematics and its applications. The analogy with certain Boltzmann factors from equilibrium statistical mechanics, already known to Dyson \cite{Dy62}, is of particular relevance to the present work. The $\beta$-ensemble, as one of the most important consequences of this perspective, provides a one-parameter families of particle systems that interpolate between eigenvalue distributions of several classical random matrix models. The joint probability density function (j.p.d.f) of the $\beta$-ensemble, which is also viewed as the Boltzmann factor of certain interacting particles with log-potentials, is specified by
\begin{align}\label{epdf}
	\rho(x_1,\cdots,x_N)=\frac{1}{Z_N} \prod_{l=1}^{N} \omega_{\beta}\left(x_{l}\right) \prod_{1 \leq j<k \leq N}\left|x_{k}-x_{j}\right|^{\beta},\quad \beta>0,
\end{align}
where $\omega_{\beta}(x)$ is a proper weight function and $Z_N$ is the corresponding partition function. 
Particular attention is paid to the specific $\beta=1,\,2,\,4$ cases for the reason that the corresponding  random matrix ensembles admit  orthogonal, unitary and symplectic symmetry respectively, as relevant to particular applications \cite{dyson62}. 

One of the most significant quantities in random matrix theory are the $k$-point correlation functions
\begin{align*}
\rho_{N,k}(x_1,\cdots,x_k)=\frac{N!}{(N-k)!}\int_{\mathbb{R}^{N-k}}\rho(x_1,\cdots,x_N)dx_{k+1}\cdots dx_N
\end{align*}
since they quantify many primary statistical quantities of an ensemble. For example, one-point correlation function represents the (particle or eigenvalue) density function of the random matrix. It is well known that for different $\beta$, these correlation functions admit different algebraic structures. In the case $\beta$ equals to $2$, it is a standard result that \cite{borodin09,mehta04}
\begin{align*}
\rho_{N,k}(x_1,\cdots,x_k)=\det\left[
K_N(x_{j_1},x_{j_2})
\right]_{j_1,j_2=1}^k,
\end{align*}
for a certain kernel $K_N(x,y)$ independent of $k$. However, for $\beta$ equals to $1$ or $4$, Pfaffian structures are involved and as such bring difficulties to the analysis of the correlation functions.

In \cite{adler00}, it was demonstrated that the correlation functions for $\beta=2$ and $\beta=1,\,4$ are inter-related. This has been exhibited by introducing the particular inner product and skew inner products respectively, namely,
\begin{align*}\displaystyle
&\beta=1:\qquad\langle \phi,\psi\rangle_1:=\frac{1}{2}\int_{-\infty}^\infty \int_{-\infty}^\infty \text{sgn}(y-x)\phi(x)\psi(y)e^{-V(x)-V(y)}dxdy,\\
&\beta=2:\qquad\langle \phi, \psi\rangle_2:=\int_{-\infty}^\infty \phi(x)\psi(x)e^{-2V(x)}dx,\\
&\beta=4:\qquad\langle \phi,\psi\rangle_4:=\frac{1}{2}\int_{-\infty}^\infty  (\phi(x)\psi'(x)-\phi'(x)\psi(x))e^{-2V(x)}dx.
\end{align*} 
For classical weights $\omega_2(x)=e^{-2V(x)}$ satisfying
\begin{align*}
	2V'(x)=-\frac{g(x)}{f(x)},\quad \text{deg $f$ $\leq2$, deg $g$ $\leq1$},
\end{align*}
one can construct the operator
\begin{align}\label{to}
\mathbf{n}:=f\frac{d}{dx}+\left(
\frac{f'-g}{2}
\right)\end{align}
and its inverse $\mathbf{n}^{-1}$ to connect the skew inner products with the inner product via
\begin{subequations}
\begin{align}
&\langle \phi, \mathbf{n}^{-1}\psi\rangle_2=\left.-\langle \phi,\psi\rangle_{1}\right|_{2V(x)\mapsto 2V(x)+\log f(x)},\label{conn1}\\
&\langle \phi,\mathbf{n}\psi\rangle_2=\left.\langle\phi,\psi\rangle_{4}\right|_{2V(x)\mapsto 2V(x)-\log f(x)}.\label{conn2}
\end{align}
\end{subequations}
From this, a relation between the classical orthogonal polynomials and the corresponding skew orthogonal polynomials can be constructed.  As a result, correlation functions of orthogonal and symplectic ensembles are obtained as a rank-one perturbation of the unitary one \cite{adler00,widom99}, thus facilitating calculations for the edge limits of orthogonal and symplectic ensembles \cite{deift07,forrester10,forrester21}. Moreover, it sheds light on the connections between Pfaff lattice and Toda lattice, which is an important result in classical integrable systems \cite{adler01}.


A new strand of random matrix theory was created by the advent of various physically motivated,  combinatorially based models which give rise to discrete analogues of \eqref{epdf}; see e.g.~\cite[Ch.~10]{forrester10}. Specific examples include the random increasing subsequence problem for a random permutation \cite{baik98}, and the polynuclear growth model \cite{prahofer00}. The discrete generalization of \eqref{epdf} in the case of unitary symmetry ($\beta = 2$) is the simplest with a determinant structure for the correlations, and has been developed to include all weight functions of the Wilson class  \cite{johansson01,johansson03,baik07} due to its relevance to applications. Applications also give rise to discrete models whose correlation functions exhibit Pfaffian structures, such as combinatorial models relating to random involutions \cite{forrester06} and certain tiling models in the non-intersecting paths picture \cite{nagao02}. These correspond to discrete generalisations of random matrix ensembles with orthogonal and symplectic symmetry. 
Subsequently, discrete orthogonal and symplectic ensembles on the linear lattice were considered by Borodin and Strahov \cite{borodin092}, in which the authors showed that the correlation kernels for discrete orthogonal and symplectic ensembles are rank-one perturbations of the discrete unitary ensemble in operator forms. Recently,  a $q$-symplectic ensemble on exponential lattice was proposed by the inspiration of the $q$-Selberg integral in \cite{forrester20}.

The functional form on a linear lattice that generalizes \eqref{epdf}, and maintains integrability properties, is hinted at in Jack polynomial theory \cite[Ch.~VI, \S 10]{macdonald79}. It is given explicitly in \cite{borodin17} by a j.p.d.f in terms of gamma functions as
\begin{align}\label{dbeta}
	\prod_{1\leq i<j\leq N}\frac{\Gamma(\ell_j-\ell_i+1)\Gamma(\ell_j-\ell_i+\theta)}{\Gamma(\ell_j-\ell_i)\Gamma(\ell_j-\ell_i+1-\theta)}\prod_{i=1}^N\omega(\ell_i;N)
\end{align}
on ordered $N$-tuples $\ell_1<\cdots<\ell_N$, such that $\ell_i=\lambda_i+\theta i$ and $\lambda_1\leq \cdots\leq \lambda_N$ are integers. Here $\ell_i$ is the positions of $N$ particles, and moreover, the ratios of gamma functions behave as $(\ell_j-\ell_i)^{2\theta}$ if we set $\ell_i=Lx_i$ and take $L\to\infty$. Therefore, it is a natural discretization of the original $\beta$-ensemble. It should be remarked here that those particles does not necessarily lie on the integer lattices since $\theta$ is not always an integer. 
In \cite{olshanski21}, Olshanski considered a Macdonald level extension, that is, $\beta$-ensemble on the exponential lattices (or so called $q$-lattices) with j.p.d.f. proportional to
\begin{align}\label{qselb}
\prod_{1\leq i\ne j\leq N}\prod_{r=0}^{\tau-1}(x_i-q^r x_j)\prod_{i=1}^N \omega(x_i).
\end{align}
The factor over pairs is familiar in the theory of $q$-generalizations of the Selberg integral, which itself contains the product over pairs as appears in (\ref{epdf}); see e.g.~the review \cite{FW08}, and point (2) in Section \ref{S1.2} below.
One notices that when $\tau=1$, the product over pairs is equal to the Vandermonde squared and thus it is a $q$-version of the unitary symmetry ensemble; and when $\tau=2$, the j.p.d.f could be formulated as
\begin{align}\label{1.5a}
\prod_{1\leq i<j\leq N} (x_i-x_j)^2(x_i-qx_j)(x_j-qx_i)\prod_{i=1}^N \omega(x_i),
\end{align}
which is the j.p.d.f of the $q$-symplectic symmetry ensemble considered in \cite{forrester20}. However, this extension doesn't contain the $q$-orthogonal symmetry ensemble  case, which is the main topic of this article.

\subsection{Motivations and main results of the article}\label{S1.2}
The motivations of this article come from two considerations.
\begin{enumerate}
\item To find a well motivated model on a $q$-lattice which limits to the $\beta = 1$ case of (\ref{epdf}). As an abuse of notation, such a model will be referred to as
a $q$-orthogonal ensemble for short, just as (\ref{1.5a}) will be referred to as a $q$-symplectic ensemble. In recent years, there have been several works discussing $q$-generalisations of the $\beta$-ensemble, for example, \cite{dimitrov19,olshanski21}.  However these particular generalisations do not contain the $\beta = 1$ case of (\ref{epdf}) as a limit.
Formulating one that does, and analysing its integrability is therefore an open problem.        
Since duality relating $\beta$ to $4/\beta$
appears in the theory of classical $\beta$ ensembles \cite{Fo92j,desrosiers09}, it seems reasonable to expect that the known formulation of the $q$-symplectic ensemble  ($\beta = 4$ case)
\cite{forrester20} and
its associated links with the theory of skew-$q$ orthogonal ensembles will have a $\beta = 1$ $q$-orthogonal ensemble counterpart.
\item $q$-Selberg integral. The Selberg integral and associated theory plays a fundamental role in the theory of the classical $\beta$ ensembles; see e.g.~\cite[Ch.~4]{forrester20}.
Likewise in relation to $q$ generalisations, one sees that the $q$-Selberg integral due to Askey-Habsieger-Kadell \cite{askey80,habsieger88,kadell88}
\begin{align}\label{q2}
\frac{1}{n!}\int_{[0,1]^n}\prod_{1\leq j<k\leq n}&\left(
(z_j-z_k)\prod_{1-\gamma\leq l\leq \gamma-1}(z_j-q^l z_k)
\right)\prod_{i=1}^n z_i^{\alpha-1}(qz_i;q)_{\beta-1}d_qz_i  \\
&=q^{\alpha\gamma{n\choose 2}+2\gamma^2{n\choose 3}}\prod_{j=1}^n \frac{\Gamma_q(\alpha+(j-1)\gamma)\Gamma_q(\beta+(j-1)\gamma)\Gamma_q(j\gamma)}{\Gamma_q(\alpha+\beta+(n+j-2)\gamma)\Gamma_q(\gamma)}\nonumber
\end{align}
has relevance as the partition function of the j.p.d.f \eqref{qselb} with weight functions chosen as the little $q$-Jacobi weight. However, it doesn't contain the $q$-orthogonal ensemble case. 
For this purpose we need knowledge of a generalisation of this $q$-integral due to Aomoto.
 In \cite{aomoto98}, Aomoto obtained the $q$-integral evaluation formula
 \begin{align}\label{q1}
\int_{z_1=0}^{1}\cdots\int_{z_n=0}^{q^\gamma z_{n-1}}&\prod_{1\leq j<k\leq n}z_j^{2\gamma-1}\frac{(q^{1-\gamma}z_k/z_j;q)_\infty}{(q^\gamma z_k/z_j;q)_\infty}(z_j-z_k)\prod_{i=1}^n z_i^{\alpha-1}(qz_i;q)_{\beta-1}
d_qz_i\\
&=q^{\alpha\gamma{n\choose 2}+2\gamma^2{n\choose 3}}\prod_{j=1}^n \frac{\Gamma_q(\alpha+(j-1)\gamma)\Gamma_q(\beta+(j-1)\gamma)\Gamma_q(j\gamma)}{\Gamma_q(\alpha+\beta+(n+j-2)\gamma)\Gamma_q(\gamma)},\nonumber
\end{align}
which holds true for arbitrary $\alpha,\,\beta,\,\gamma\in\mathbb{C}$ satisfying $|q^{\alpha+(i-1)\gamma}|<1$ when $i=1,\,\cdots,\, n$. It was shown in \cite{ito17} that upon an
appropriate modification of Aomoto's integral that it does reduce to the $q$-integral of Askey-Habsieger-Kadell when $\gamma$ is a positive integer.
Moroever, when $\gamma=1/2$ it limits for $q \to 1$ to the $\beta = 1$ case of (\ref{epdf}) with a Jacobi weight. This motivates us to take the product over differences
in Aomoto's integral when $\gamma=1/2$
as specifying the sought $q$-orthogonal ensemble.
This definition shares the property of the orthogonal symmetry linear lattice in (\ref{dbeta}) that not
all particles are restricted to the integer lattice, rather in the $q$-orthogonal ensemble case they are restricted to the configuration space 
\begin{align}\label{cs}
 \{(z_1,\cdots,z_n)\,|\, 0<q^{-(n-1)/2}z_n<\cdots<z_1<1\}.
 \end{align}
 Here $z_1,z_3,\dots$ take values on the integer $q$ lattice while $z_2,z_4,\dots$ are restricted to the half integer $q$ lattice.
\end{enumerate}

\subsection{Organization of the article}
This article is organized as follows. In Section \ref{sec2}, we start with Aomoto's $q$-Selberg integral and extract the model
\begin{align}\label{q1a}
\prod_{1\leq i<j\leq n} (z_i-z_j)\prod_{i=1}^n \omega(z_i;q)
\end{align}
with configuration space defined by \eqref{cs}. However, with the configuration space asymmetric, it is hard for us for further analysis. Therefore, from the discrete version of de Bruijn formula, we rewrite Aomoto's $q$-Selberg integral into a Pfaffian form and a skew inner product is formulated. 
One key point in Section \ref{sec2} is to introduce a specific skew-symmetric kernel, by which we make the configuration space symmetric. The alternative $q$-orthogonal ensemble is formulated by
Definition \ref{prop4}, which is the j.p.d.f we mainly discuss in this paper.
Then we introduce a family of skew $q$-orthogonal polynomials, which are the average characteristic polynomials of this ensemble. In Section \ref{sec3}, the connection between $q$-skew inner product and $q$-inner product is made by discrete Pearson relation, and it is clarified by 
a $q$-analog of the operator $\mathbf{n}$ in \eqref{to} and its inverse. Thus we achieve the goal relating classical skew $q$-orthogonal polynomials to classical $q$-orthogonal polynomials. Several classical examples of skew $q$-orthogonal polynomials are given in Section \ref{sec4}, including the little $q$-Jacobi case, the Al Salam \& Carlitz case, the $q$-Laguerre case and the big $q$-Jacobi case. As a result, some multiple integrals with classical $q$-weights are obtained by realising that the multiple integrals are the normalizations of the classical skew orthogonal polynomials. Since classical skew $q$-orthogonal polynomials are connected with $q$-orthogonal polynomials, one can compute the correlation function of $q$-orthogonal ensemble as a rank one perturbation of classical $q$-unitary ensemble . This result is given in Section \ref{sec5} and examples of correlation functions for certain classical $q$-orthogonal ensembles are given.
Some further remarks are given in Section \ref{sec6}.

\section{Aomoto's $q$-Selberg intergal and $q$-SOPs in $\beta=1$ ensemble}\label{sec2}
In this part, we start with the $q$-Selberg integral (\ref{q2}) given by Askey-Habsieger-Kadell.
 When $\gamma=1$, this integration formula suggests the j.p.d.f
\begin{align}\label{ujpdf}
\prod_{1\leq j<k\leq n}(z_j-z_k)^2\prod_{i=1}^n \omega(z_i;q)d_qz_i,
\end{align}
with $\{z_i\}_{i=1}^n$ in the configuration space $X$ of exponential lattices.  The latter is a subspace of the quantization of the real line
\begin{align*}
\mathbb{L}:=\mathbb{L}_-\cup \mathbb{L}_+\subset\mathbb{R},\quad \mathbb{L}_\pm:=\xi_\pm q^{\mathbb{Z}}
\end{align*}
with $\xi_\pm\in\mathbb{R}$ and $q^{\mathbb{Z}}:=\{q^n|n\in\mathbb{Z}\}$.
The particles $\{z_i\}_{i=1}^n$ are not necessarily in order and $\omega(z_i;q)$ is a weight function supported on this configuration space. The Askey-Hasbieger-Kadell's formula \eqref{q2} suggests the special weight $\omega(z;q)=z^{\alpha-1}(qz;q)$, which is the little $q$-Jacobi weight and can be taken as representative of more general
 classical weights underlying the Askey scheme. For more example, one can refer to \cite{olshanski21} for big $q$-Jacobi case. The form of \eqref{ujpdf} can be regarded as a generalised j.p.d.f of an Hermitian  matrix model with unitary invariance, whose eigenvalues are supported on the exponential lattices. Therefore, it can be analysed in the framework of determinantal point process \cite{borodin09}. 

There is another structured random matrix ensemble arising from (\ref{q2}).
Thus when $\gamma=2$, the j.p.d.f (\ref{1.5a})
is extracted, and this is a $q$-generalization of the symplectic invariant ensemble as specified by (\ref{epdf}) with $\beta = 4$.
 In the recent work \cite{forrester20}, the integrability of such $q$-symplectic ensembles with weight functions from the Askey scheme is given in terms of the correlation function, which is a rank one perturbation of the $q$-unitary ensemble, and certain examples were  explicitly given.

The requirement $\gamma\in\mathbb{Z}_+$ prohibits  \eqref{q2}  relating to (\ref{epdf}) with $\beta = 1$, which is the orthogonal symmetry case.
For this we turn to Aomoto's integral (\ref{q1}).
The special choice $\gamma = 1/2$ allows us to extract the j.p.d.f (\ref{q1a}). The particle coordinates must satisfy the inequalities
(\ref{cs}).
More explicitly, according to the definition of
Jackson's $q$-integral \cite{kaneko96}
\begin{align*}
\int_0^1\cdots\int_{z_n=0}^{\qf z_{n-1}}f(z_1,\cdots,z_n)\prod_{j=1}^n\frac{d_qz_j}{z_j}:=(1-q)^n\sum_{\langle \xi_F\rangle}f(t_1,\cdots,t_n)
\end{align*}
with the summation region $\langle\xi_F\rangle$ 
\begin{align}\label{cs1}
t_1=q^{\mu_1},\, t_2/t_1=q^{\mu_2+1/2},\,\cdots,\, t_n/t_{n-1}=q^{\mu_n+1/2},\quad \mu_1,\,\cdots,\,\mu_n\in \mathbb{Z}_{\geq0}.
\end{align}
Therefore, for general weights, one can define the following $q$-orthogonal ensemble. 
\begin{definition}\label{def1}
The discrete $q$-orthogonal ensemble with general weight can be defined by
\begin{align}\label{dist}
\prod_{1\leq i<j\leq n} (z_i-z_j)\prod_{i=1}^n \omega(z_i;q),
\end{align}
with configuration space 
\begin{align*}
\{(z_1,\cdots,z_n)\,|\, a<q^{-(n-1)/2}z_n<\cdots<z_1<b\}, \quad a,b\in\mathbb{R}.
\end{align*}
Moreover, for $(a,b) = (0,1]$, the
 $\{z_i\}_{i=1}^n$ obey the same relations as the $\{t_i\}_{i=1}^n$ in (\ref{cs1}).
\end{definition}

To display the algebraic structure for this model, we make use of a $q$-version of the de Bruijn formula \cite{debruijn55} to write 
the corresponding partition function  in terms of a Pfaffian. For simplicity, we assume that the weight is supported on the interval $(a,b) = (0,1]$.
Firstly we suppose $n=2m$ is even.  According to the Vandermonde determinant formula 
\begin{align*}
\prod_{1\leq i<j\leq n}(t_i-t_j)=(-1)^{\frac{n}{2}}\det(t_i^{j-1})_{i,j=1}^n.
\end{align*}
Summing over $t_1$ for the first row gives (where $\tilde{\omega}(x)=x\omega(x;q)$)
\begin{align*}
(1-q)\left[
\sum_{\mu_1=0}^{\bar{\mu}_1}\tilde{\omega}(q^{\mu_1}),\,\sum_{\mu_1=0}^{\bar{\mu}_1}q^{\mu_1}\tilde{\omega}(q^{\mu_1}),\,\cdots,\,\sum_{\mu_1=0}^{\bar{\mu}_1}q^{(n-1)\mu_1}\tilde{\omega}(q^{\mu_1})
\right]
\end{align*}
with $\bar{\mu}_1$ defined by $q^{\bar{\mu}_1}:=t_2q^{-1/2}$. If we sum over all odd numbered rows we can see the summation is over $\sum_{\hat{\mu}_{n-1}}^{\bar{\mu}_{n-1}}$ with $\hat{\mu}_{n-1}$ and $\bar{\mu}_{n-1}$ defined by $q^{\hat{\mu}_{n-1}}=t_{n-2}\qf $ and $q^{\bar{\mu}_{n-1}}=t_nq^{-1/2}$ respectively. We observe that adding row $1$ to row $3$ allows the row $3$ to be written as $\sum_{\mu_3=0}^{\bar{\mu}_3}$.  Proceeding in this way, all odd numbered rows can be chosen to begin from $\mu=0$, and the final such row, numbered $n-1$, is then over $\sum_{\mu_{n-1}=0}^{\bar{\mu}_{n-1}}$.
The sum over the even labelled variables is parametrised by $t_{2j}=\qf q^{\mu_{2j}}$ with $0\leq \mu_2<\cdots<\mu_{2j}$.
Moreover, the determinant is symmetric in the $\mu_{2j}$ so we can sum over $0\leq \mu_{2j}<\infty$ for $j=1,\cdots,n/2$ provided we divide by $(n/2)!$. Therefore, as a result of de Bruijn formula \cite{debruijn55}, these facts lead us to the Pfaffian expression
\begin{align}\label{deb}
\begin{aligned}
\int_{z_1=0}^{1}\cdots\int_{z_{2m}=0}^{\qf z_{2m-1}}&\prod_{1\leq i<j\leq 2m}(z_i-z_j)\prod_{i=1}^{2m}\omega(z_i;q)d_qz_i\\
&=\Pf\left[
\int_0^1\int_0^{\qf y}(x^iy^j-x^jy^i)\omega(x;q)\omega(y;q)d_qxd_qy
\right]_{i,j=0}^{2m-1}.
\end{aligned}
\end{align}
In the case that $ n=2m+1 $ is odd, similar calculations following the strategy of \cite{debruijn55} show
\begin{align*}
	\int_{z_1=0}^{1}\cdots&\int_{z_{2m+1}=0}^{\qf z_{2m}}\prod_{1\leq i<j\leq 2m+1}(z_i-z_j)\prod_{i=1}^{2m+1}\omega(z_i;q)d_qz_i=\Pf\left[\begin{array}{cc}
	[m_{j,k}]& [\xi_j]\\
	-[\xi_k]^\top& 0
	\end{array} \right]_{j,k=0}^{2n},
\end{align*}
where bi-moments $\{m_{j,k}\}$ and single moments $\{\xi_j\}$ are defined by
\begin{align}\label{ms}
m_{j,k}=\int_0^1\int_0^{\qf y}(x^jy^k-x^ky^j)\omega(x;q)\omega(y;q)d_qxd_qy,\quad \xi_{j}=\int_0^1 x^{j}\omega(x;q)d_qx.
\end{align}
We see from the elements of the Pfaffian in both the even and odd cases that underpinning these algebraic expressions is
the particular skew-symmetric inner product
\begin{align}\label{sip}
\langle \phi(x),\psi(y)\rangle_{1,\omega}&=\int_0^1\int_0^{\qf y} \left(\phi(x)\psi(y)-\phi(y)\psi(x)\right)\omega(x;q)\omega(y;q)d_qxd_qy\\
&=(1-q)^2\sum_{\gamma=0}^\infty\sum_{\kappa=0}^\gamma q^{\gamma+\kappa+1/2}\left(
\phi(q^{\gamma+1/2})\psi(q^\kappa)-\phi(q^\kappa)\psi(q^{\gamma+1/2})
\right)\omega(q^{\gamma+1/2};q)\omega(q^\kappa;q),\nonumber
\end{align}
which can be taken as a characterization of the  $q$-orthogonal ensemble.

\subsection{An alternative definition of $q$-orthogonal ensemble in symmetric configuration space}
As remarked, the $q$-orthogonal ensemble induced by Aomoto's $q$-Selberg integral formula is defined on a special configuration space, which is asymmetric and distinguished by both half integer and integer lattices.
The aim of this subsection is to construct a skew symmetric kernel $s(x,y)$, such that the configuration space is  symmetric. We start with an identity relating to the iterated $q$-integral
in (\ref{sip}).

\begin{proposition}
If we take
\begin{align}\label{ssk}
s(q^{\frac{i}{2}},q^{\frac{j}{2}})=(1+\qf )^2 \left\{\begin{array}{lc}
1,&\text{if $ i=2m+1,j=2n $ and $ m\geq n $},\\
-1,&\text{if $ i=2n,j=2m+1 $ and $ m\geq n $},\\
0,&\text{other cases},
\end{array}
\right.
\end{align}
then
\begin{align*}
\int_0^1\int_0^1 \phi(x)\psi(y)s(x,y)\omega(x;q)\omega(y;q)d_{\qf }xd_{\qf }y=\int_0^1\int_0^{\qf y}(\phi(x)\psi(y)-\phi(y)\psi(x))\omega(x;q)\omega(y;q)d_qxd_qy.
\end{align*}
\end{proposition}
\begin{proof}
To prove this, we first write out the skew symmetric inner product
\begin{align*}
\int_0^1\int_0^1 \phi(x)\psi(y)s(x,y)\omega(x;q)\omega(y;q)d_qxd_qy,\quad s(x,y)=-s(y,x),
\end{align*}
which is equivalent to the double sum expression
\begin{align*}
(1-q)^2 \sum_{i,j=0}^\infty \phi(q^i)\psi(q^j) s(q^i,q^j)\omega(q^i;q)\omega(q^j;q)q^{i+j}.
\end{align*}
To coincide with the skew inner product \eqref{sip}, we should scale $q$ to $\qf $ and obtain 
\begin{align*}
(1-\qf )^2\sum_{i,j=0}^\infty \phi(q^{\frac{i}{2}})\psi(q^{\frac{j}{2}})s(q^{\frac{i}{2}},q^{\frac{j}{2}})\omega(q^{\frac{i}{2}};q)\omega(q^{\frac{j}{2}};q)q^{\frac{i+j}{2}}.
\end{align*}
Comparing with the series form in  \eqref{ssk}, one can verify the stated identity.
\end{proof}
Use of the de Bruijn formula is now made backwards, and we obtain
\begin{align*}
&\Pf \left[
\int_0^1\int_0^1 x^iy^js(x,y)\omega(x;q)\omega(y;q)d_{\qf }xd_{\qf }y
\right]_{i,j=0}^{2n-1}\\
&\qquad\qquad= \frac{1}{(2n)!}\int_{[0,1]^{2n}}\Pf\left[s(z_i,z_j)
\right]_{i,j=1}^{2n}\prod_{1\leq i<j\leq 2n}(z_i-z_j)\prod_{i=1}^{2n}\omega(z_i;q)d_{\qf }z_i.
\end{align*}
Therefore, an alternative j.p.d.f specifying the $q$-orthogonal ensemble is 
\begin{align*}
\Pf\left[s(z_i,z_j)\right]_{i,j=1}^{2n}\prod_{1\leq i<j\leq 2n}(z_i-z_j)\prod_{i=1}^{2n}\omega(z_i;q)d_{\qf }z_i,
\end{align*}
where $s(x,y)$ is a skew symmetric kernel defined on $q$-lattices as \eqref{ssk} and the configuration space is
\begin{align*}
\left\{
(z_1,\cdots,z_{2n})|z_i\in q^{\mathbb{\mathbb{N}}/2}
\right\}.
\end{align*}
We remark that the skew symmetric kernel $s(x,y)$ has a similar representation in the linear discrete case as given by Borodin and Strahov in their proof of \cite[Lemma 6.1]{borodin09}.

In a similar manner, if we define
\begin{align}\label{F}
	F(x):=\left\{\begin{array}{ll}
	0,&x\in q^{\mathbb{Z}+\frac{1}{2}},\\
	1+q^{\frac{1}{2}},&x\in q^{\mathbb{Z}},
	\end{array}\right.
\end{align}
then we have
\begin{align*}
	\int_0^1 g(x)\omega(x;q)d_qx=\int_{0}^{1}g(x)F(x)\omega(x;q)d_{\qf }x.
\end{align*}
Thus, by the de Bruijn formula, we have
\begin{align*}
	\Pf\left[\begin{array}{cc}
	[m_{j,k}]& [\xi_j]\\
	-[\xi_k]^\top& 0
	\end{array} \right]_{j,k=0}^{2n}=\int_{[0,1]^{2n+1}}\Pf\left[
	\begin{array}{cc}
	[s(z_i,z_j)] & [F(z_i)]\\
	-[F(z_j)]^\top& 0
	\end{array}
	\right]_{i,j=1}^{2n+1}\Delta_{2n+1}(z)\prod_{i=1}^{2n+1}\omega(z_i;q)d_{\qf }z_i,\\
\end{align*}
where the moments $\{m_{j,k}\}$ and $\{\xi_j\}$ are given in \eqref{ms} and $\Delta_{2n+1}(z)$ is the Vandermonde product $\prod_{1\leq i<j\leq 2n+1}(z_i-z_j)$.
Therefore we get the j.p.d.f  in the odd case
\begin{align*}
	\Pf\left[
	\begin{array}{cc}
	[s(z_i,z_j)]& [F(z_i)]\\
	-[F(z_j)]^\top & 0
	\end{array}
	\right]_{i,j=1}^{2n+1} \Delta_{2n+1}(z)\prod_{i=1}^{2n+1}\omega(z_i;q).
\end{align*}

These finding are summarised in the following.

\begin{proposition}\label{prop4}
At the level of the partition function, we have an alternative j.p.d.f for the $q$-orthogonal ensemble, defined by 
\begin{align*}
\Pf\left[s(z_i,z_j)\right]_{i,j=1}^{2n}\Delta_{2n}(z)\prod_{i=1}^{2n}\omega(z_i;q)d_{\qf }z_i,
\end{align*}
for even number of particles and 
\begin{align*}
	\Pf\left[
	\begin{array}{cc}
	[s(z_i,z_j)]& [F(z_i)]\\
	-[F(z_j)]^\top & 0
	\end{array}
	\right]_{i,j=1}^{2n+1} \Delta_{2n+1}(z)\prod_{i=1}^{2n+1}\omega(z_i;q)d_{q^\frac{1}{2}}z_i
\end{align*}
for odd number of particles, where $s(x,y)$ and $F(x)$ are defined in \eqref{ssk} and \eqref{F} respectively. Moreover, the configuration space is given by 
\begin{align*}
\left\{(z_1,\cdots,z_n)|z_i\in \{aq^{\mathbb{N}/2}, \, bq^{\mathbb{N}/2}\}\right\},\quad a,\,b\in\mathbb{R}.
\end{align*}
\end{proposition}

\begin{remark}
We have shown the equivalence of the skew inner product and partition function of Aomoto's $q$-Selberg integral, and the $q$-integral defined on the symmetric configuration space. It should be remarked that an analogous skew inner product  on the  discrete linear lattice was obtained in a study of vicious random walkers \cite{nagao02}. It is given there by 
\begin{align*} 
\langle f(x), g(y)\rangle=\frac{1}{2}\sum_{y=-\infty}^\infty \sum_{x=-\infty}^\infty s(x,y)f(x)g(y)\omega(x)\omega(y),
\end{align*}
where 
\begin{align*}
s(x,y)=\left\{\begin{array}{ll}
1,& x>y\\
-1,&y>x\\
0,&y=x
\end{array}\right.
\end{align*}
and $\omega(x)$ is related to the symmetric Hahn polynomial weight. It is an interesting question to find a combinatorial model for the Young tableaux with weight $q^{|\lambda|}$ \cite{bogoliubov14,stanley96} which is related to the discrete $q$-orthogonal ensemble as formulated in the present work.
\end{remark}

\subsection{$q$-skew orthogonal polynomials induced from \eqref{sip}}
It is well known  how to construct a family of skew orthogonal polynomials from a given skew inner product via the so-called skew Borel decomposition \cite{adler99,chang18,forrester20}. In this part,  we give a brief review of this theory to illustrate how to find skew orthogonal polynomials from \eqref{sip}.

Let's denote 
\begin{align*}
m_{i,j}:=\langle x^i, y^j\rangle_{1,w}=\int_0^1\int_0^1 x^iy^js(x,y)\omega(x;q)\omega(y;q)d_{\qf }xd_{\qf }y.
\end{align*}
and define the skew-symmetric moment matrix $ M=[m_{i,j}]_{i,j\in \mathbb{N}} $. 
From the skew Borel decomposition \cite{adler99}, one knows that \begin{align}\label{sbd}
M=S^{-1}JS^{-\top}.
\end{align}
Here $S$ is a lower triangular matrix with diagonals $1$ and $J$ is a block diagonal matrix of the form
\begin{align*}
J=\left(\begin{array}{ccccc}
0&u_0&&&\\
-u_0&0&&&\\
&&0&u_1&\\
&&-u_1&0&\\
&&&&\ddots
\end{array}
\right)
\end{align*}
with $u_i=\tau_{2(i+1)}/\tau_{2i}$, where $\tau_{2i}$ is the Pfaffian of the $2i\times 2i$ principle minor of $M$ and $\tau_0=1$.

Therefore, by setting 
\begin{align*}
Q(x)=(Q_0(x),Q_1(x),\cdots)^\top=S\chi(x),\quad \chi(x)=(1,x,x^2,\cdots)^\top,
\end{align*}
one could easily demonstrate that the family of polynomials $\{Q_n(x)\}_{n\in\mathbb{N}}$ satisfy the  skew orthogonality relation
\begin{align}\label{sor}
\langle Q_{2m}(x),Q_{2n}(y)\rangle_{1,\omega}=\langle Q_{2m+1}(x),Q_{2n+1}(y)\rangle_{1,\omega}=0,\quad \langle Q_{2m}(x),Q_{2n+1}(y)\rangle_{1,\omega}=u_m\delta_{n,m}
\end{align}
with nonzero normalization factors $\{u_m\}_{m\in\mathbb{N}}$. From \cite[Sec.2]{chang18}, one knows that the linear system \eqref{sor} leads to the  Pfaffian expressions
\begin{align*}
&Q_{2m}(z)=\frac{1}{\tau_{2m}}\Pf(0,\cdots,2m,z),\quad \tau_{2m}=\Pf(0,\cdots,2m-1),\\
&Q_{2m+1}(z)=\frac{1}{\tau_{2m}}\Pf(0,\cdots,2m-1,2m+1,z),
\end{align*}
where the Pfaffian elements are defined by $\Pf(i,j)=m_{i,j}$ and $\Pf(i,z)=z^i$. Moreover, by substituting the Pfaffian expressions into the skew orthogonality relation, one finally finds $u_{m}=\tau_{2m+2}/\tau_{2m}$.
From (\ref{sor})  $Q_{2m+1}(z)\mapsto Q_{2m+1}(z)+\alpha Q_{2m}(z)$ is still skew orthogonal, allowing for the coefficient of $z^{2m}$ to be chosen to equal 0.

According to the formula \eqref{deb}, we know that $\tau_{2n}$ in fact is the partition function of the $q$-orthogonal ensemble considered. Therefore, if we know exactly the classical skew orthogonal polynomials from the $q$-orthogonal polynomials, then the partition function of the $q$-orthogonal ensemble with classical Askey scheme weights can be explicitly computed, which is similar to the computation of Catalan-Hankel Pfaffian in $q$-symplectic case \cite{shen21}.

\section{$q$-OPs, $q$-SOPs and inverse operator}\label{sec3}
In what follows, we set the integrand interval to be $(0,\infty)$ (i.e. $z_i\in q^{\mathbb{Z}/2}$ case) in general but fix the interval when specific classical weights are taken into account.
\subsection{Scale of $q$-OPs}
The inner product $\langle \cdot,\cdot\rangle_{2,\rho}$ in this paper is a variation of the standard one in the $q$ setting since we need to scale $q\to \qf $. 
For this purpose we define the  ${\qf }$-integral  (c.f. \cite[eq. (11.4.3)]{ismail05})
\begin{align*}
\int_0^b \phi(x)d_{\qf }x:=\sum_{n=0}^\infty [bq^{\frac{n}{2}}-bq^{\frac{n+1}{2}}]\phi(bq^{\frac{n}{2}})=(1-\qf )\sum_{n=0}^\infty bq^{\frac{n}{2}}\phi(bq^{\frac{n}{2}}).
\end{align*}
When $b=\infty$, the above definition can be extended to
\begin{align*}
\int_0^\infty \phi(x)d_{\qf }x:=(1-\qf )\sum_{n\in\mathbb{Z}}q^{\frac{n}{2}}\phi(q^{\frac{n}{2}}).
\end{align*}
Similarly, let us define the $\qf $-difference operator
\begin{align*}
(D_{\qf }\phi)(x)=\frac{\phi(x)-\phi(\qf  x)}{(1-\qf )x}.
\end{align*}

Consider the original monic $q$-OPs which satisfy the orthogonality relation
\begin{align}\label{3.0}
h_n(q)\delta_{n,m}=\int_0^\infty p_n(x;q)p_m(x;q)\rho(x;q)d_qx=(1-q)\sum_{s\in\mathbb{Z}}p_n(q^s;q)p_m(q^s;q)\rho(q^s;q)q^s,
\end{align}
where $ h_n(q) $ is the normalization constant.
If we scale $q\mapsto \qf $, then we have
\begin{align}\label{definition}
\langle p_n(x;\qf ),p_m(x;\qf )\rangle_{2,\rho}:=\int_0^\infty p_n(x;\qf)p_m(x;\qf)\rho(x;\qf)d_{\qf }x:=h_n(\qf )\delta_{n,m}.
\end{align}
Thus in this case, the $q$-OPs are defined in the lattices $\{q^{\frac{s}{2}}\,|\,s\in\mathbb{Z}\}$ and their expressions are scalings of $q$ to $\qf $. For example, the monic little $q$-Jacobi polynomials in this case are given by 
\begin{align*}
p_n^{(\al,\be)}(x;\qf )=(-1)^nq^{\frac{n(n-1)}{4}} \frac{(q^{\frac{\al+1}{2}};\qf )_n}{(q^{\frac{\al+\be+n+1}{2}};\qf )_n}\tilde{p}_n^{(\al,\be)}(x;\qf ),
\end{align*}
where $\tilde{p}_n^{(\al,\be)}(x;\qf )$ is the original little $q$-Jacobi polynomial \cite{masuda91}
\begin{align*}
\tilde{p}_n^{(\al,\be)}(x;\qf )=\sum_{r\geq0}\frac{(q^{-\frac{n}{2}};\qf )_r(q^{\frac{\al+\be+n+1}{2}};\qf )_r}{(\qf ;\qf )_r(q^{\frac{\al+1}{2}};\qf )_r}(\qf  x)^r.
\end{align*}
The orthogonality relation then reads
\begin{align}\label{lqj}
&\int_0^1 p_n^{(\al,\be)}(x;\qf )p_m^{(\al,\be)}(x;\qf )x^\al(\qf  x;\qf )_\be d_{\qf }x\\
&\quad=\delta_{n,m}q^{\frac{n(n+\al)}{2}}[\al+\be+2n+1]_{\qf }^{-1}\frac{(\qf ;\qf )_n(\qf ;\qf )_{n+\al}(\qf ;\qf )_{n+\be}(\qf ;\qf )_{n+\al+\be}}{(\qf ;\qf )_{2n+\al+\be}^2}.\nonumber
\end{align}
Moreover, the classical $q$-weight $\rho(x;\qf )$ satisfies the Pearson equation \cite{nikiforov86}
\begin{align}\label{pearson}
\frac{\rho(\qf x;\qf )}{\rho(x;\qf )}=\frac{f(x;\qf )-(1-\qf )xg(x;\qf )}{f(\qf x;\qf )},\quad x\in q^{\mathbb{Z}/2}
\end{align}
where the degrees of $f$ and $g$ are less or equal to $2$ and $1$ respectively. The polynomials
$(f,g)$ are referred to as a Pearson pair.

\subsection{Scale of $q$-SOPs when $\beta=4$}
In this part, we give a brief review of $q$-SOPs arising from the $q$-symplectic invariant ensemble and scale $q$ to $\qf $ for later use. Let us define the skew symmetric inner product for $\beta=4$ with weight function $\tilde{\omega}(x;\qf)$ as
\begin{align*}
\langle \phi(x;\qf),\psi(x;\qf )\rangle_{4,\tw}:=\int_0^\infty \left(
\phi(x;\qf )D_{\qf }\psi(x;\qf )-\psi(x;\qf )D_{\qf }\phi(x;\qf )
\right)\tw(x;\qf )d_{\qf }x.
\end{align*}
By defining the operator
\begin{align}\label{ma}
\mathcal{A}_q=g(x;\qf )T_{\qf }+q^{-1/2}f(x;\qf )D_{q^{-\frac{1}{2}}}+f(x;\qf )D_{\qf },
\end{align}
with $T_{\qf}\phi(x)=\phi(\qf x)$, 
one can demonstrate that \cite{forrester20}
\begin{align}\label{tw}
\langle \phi(x),\mathcal{A}_q\psi(x)\rangle_{2,\rho}=\langle \phi(x),\psi(x)\rangle_{4,\tw}.
\end{align}
Here the weight $\tw(x,\qf)$ relates to the classical weight $\rho(x;\qf)$ via the Pearson pair \eqref{pearson} according to the formula
\begin{align}\label{tw1}
 \tw(x;\qf )=f(\qf x;\qf )\rho(\qf x;\qf ),
\end{align}
with assumption that $\tw(x;\qf)$ vanishes at the end points of the supports.
Moreover, the operator $\mathcal{A}_q$ enjoys the property
\begin{align*}
\langle \mathcal{A}_q\phi(x),\psi(x)\rangle_{2,\rho}=-\langle\phi(x),\mathcal{A}_q\psi(x)\rangle_{2,\rho},
\end{align*}
which leads to the structural equation 
\begin{align}\label{rec1}
\mathcal{A}_qp_n(x;\qf )=-\frac{c_n}{{h}_{n+1}(\qf )}p_{n+1}(x;\qf )+\frac{c_{n-1}}{h_{n-1}(\qf )}p_{n-1}(x;\qf ).
\end{align}
For future reference we note too the weaker result
\begin{align}\label{rec}
\mathcal{A}_qp_k(x;\qf )=\sum_{\ell=0}^{N-1}\frac{\langle p_\ell,\mathcal{A}_qp_k\rangle_{2,\rho}}{\langle p_\ell,p_\ell\rangle_{2,\rho}}p_\ell(x;\qf ),\quad k\ne N-1
\end{align}
according to the orthogonality.
\subsection{The construction of inverse operator}
The aim is to construct an operator $\mb_q$ such that 
\begin{align}\label{21rel}
\langle \phi(x),\mb_q \psi(x)\rangle_{2,\rho}=-\langle \phi(x),\psi(x)\rangle_{1,\omega},
\end{align}
where $\langle \cdot,\cdot\rangle_{1,\omega}$ was defined by \eqref{sip}.
By definition \eqref{definition}, we know that
\begin{align}\label{inn}
\begin{aligned}
\langle \phi(x),\mb_q\psi(x)\rangle_{2,\rho}&=\int_0^\infty \phi(x)\mb_q\psi(x)\rho(x)d_{\qf }x=(1-\qf )\sum_{s=-\infty}^\infty\phi(q^{\frac{s}{2}})\mb_q\psi(q^{\frac{s}{2}})\rho(q^{\frac{s}{2}})q^{\frac{s}{2}}.
\end{aligned}
\end{align}
Now, we want to demonstrate that if $\mb_q$ has an explicit relation
\begin{align}\label{relba}
\mb_q=(1+\qf )^2\mathcal{A}_q^{-1},
\end{align}
 with the operator $\mathcal{A}_q$ defined in \eqref{ma}, 
then \eqref{21rel} will be satisfied.
To this end, 
we first  rewrite $\mathcal{A}_q$ with the help of the Pearson equation \eqref{pearson},
\begin{align}\label{relar}
\mathcal{A}_q=\frac{\rho^{-1}(x)}{(1-\qf )x}\left[
f(x)\rho(x)T_{q^{-\frac{1}{2}}}-f(\qf x)\rho(\qf x)T_{\qf }
\right]:=\frac{\rho^{-1}(x)}{(1-\qf )x}\mathcal{R}_q,
\end{align}
and then proceed to specify  the inverse operator of $\mathcal{R}_q$.
\begin{proposition}
The inverse operator of $\mathcal{R}_q$ can be written as $\epsilon_q$ acting on half integer and integer lattices, and it admits the forms
\begin{align}\label{epsilon}
\begin{aligned}
&(\epsilon_q\cdot \psi)(q^m)=\sum_{k=m}^\infty \frac{\tw(q^{m+1/2})\cdots \tw(q^{k-1/2})}{\tw(q^{m})\cdots \tw(q^{k})}\psi(q^{k+1/2}),\\
&(\epsilon_q\cdot\psi)(q^{m+1/2})=-\sum_{k=-\infty}^m\frac{\tw(q^{k+1/2})\cdots \tw(q^{m-1/2})}{\tw(q^{k})\cdots \tw(q^{m})}\psi(q^k),
\end{aligned}
\end{align}
where $\tw(x)=\tw(x;\qf )$ is given by \eqref{tw1}.
\end{proposition}
\begin{proof}
Here we provide the proof which contains the $ q^{\mathbb{Z}} $ part only, the $ q^{\mathbb{Z}+\frac{1}{2}} $ case can be verified by a similar calculation. Noting that
\begin{align*}
&\mathcal{R}_q[(\epsilon_q\cdot \psi)(q^m)]=\left.[\tw(q^{-\frac{1}{2}}x)\epsilon_q\cdot\psi(q^{-\frac{1}{2}}x)-\tw(x)\epsilon_q\cdot \psi(\qf  x)]\right|_{x=q^m}\\
&\quad=-\sum_{k=-\infty}^{m-1}\frac{\tw(q^{k+\frac{1}{2}})\dots\tw(q^{m-\frac{1}{2}})}{\tw(q^k)\dots\tw(q^{m-1})}\psi(q^k)+\sum_{k=-\infty}^{m}\frac{\tw(q^{k+\frac{1}{2}})\dots\tw(q^{m-\frac{1}{2}})}{\tw(q^k)\dots\tw(q^{m-1})}\psi(q^k)=\psi(q^m).
\end{align*}
On the other hand, one can easily check that
\begin{align*}
&\epsilon_q[\mathcal{R}_q\cdot \psi(q^{m})]=\left.\epsilon_q\cdot[\tw(q^{-\frac{1}{2}}x)\psi(q^{-\frac{1}{2}}x)-\tw(x)\psi(\qf  x)]\right|_{x=q^m}\\
&\quad=\sum_{k=m}^{\infty}\frac{\tw(q^{m+\frac{1}{2}})\dots \tw(q^{k-\frac{1}{2}})}{\tw(q^m)\dots\tw(q^{k-1})}\psi(q^k)-\sum_{k=m}^\infty\frac{\tw(q^{m+\frac{1}{2}})\dots\tw(q^{k+\frac{1}{2}})}{\tw(q^m)\dots\tw(q^k)}\psi(q^{k+1})\\
&\quad=\psi(q^m).
\end{align*}
\end{proof}

\begin{coro}
We have the following relation between skew inner product \eqref{sip} and inner product \eqref{definition} via
\begin{align}\label{ipr}
\langle \phi(x),\mb_q \psi(x)\rangle_{2,\rho}=-\langle \phi(x),\psi(x)\rangle_{1,\omega}
\end{align}
with $\mb_q$ given by \eqref{relba}. Moreover, the weight function $\omega(x;q)$ in the skew inner product satisfies
\begin{align}\label{newweight}
\omega(x;q)\omega(\qf x;q)=\frac{\rho(x;\qf )}{f(\qf x;\qf )},\quad x\in q^{\mathbb{Z}/2}.
\end{align}
\end{coro}

\begin{proof}
Firstly, we have
\begin{align*}
\big\langle \phi(x),\mb_q\psi(x)\big\rangle_{2,\rho}=(1-q)(1+\qf )\big\langle \phi(x),\mathcal{R}_q^{-1}(x\rho(x)\psi(x))\big\rangle_{2,\rho}
\end{align*}
by the help of eqs. \eqref{relba} and \eqref{relar}. Now $\mathcal{R}_q^{-1}$ in the above equation can be replaced by $\epsilon_q$ and by expanding the inner products, one gets
\begin{align*}
(1-q)^2\sum_{k=-\infty}^\infty &\sum_{m=-\infty}^k \frac{f(q^{m+1};\qf )\rho(q^{m+1};\qf )\cdots f(q^k;\qf )\rho(q^k;\qf )}{f(q^{m+\frac{1}{2}};\qf )\rho(q^{m+\frac{1}{2}};\qf )\cdots f(q^{k+\frac{1}{2}};\qf )\rho(q^{k+\frac{1}{2}};\qf )}\\
&\times \rho(q^m)\rho(q^{k+\frac{1}{2}})q^{k+m+\frac{1}{2}}\left(
\phi(q^m)\psi(q^{k+\frac{1}{2}})-\phi(q^{k+\frac{1}{2}})\psi(q^m)
\right).
\end{align*}
If we introduce a new weight $\omega(x;q)$ as in \eqref{newweight}, 
then the above double series sum can be reformulated as
\begin{align*}
(1-q)^2\sum_{k=-\infty}^\infty \sum_{m=-\infty}^k \omega(q^m){\omega}(q^{k+\frac{1}{2}})q^{k+m+\frac{1}{2}}\left(\phi(q^m)\psi(q^{k+\frac{1}{2}})-\phi(q^{k+\frac{1}{2}})\psi(q^m)
\right)
\end{align*}
which is indeed the skew inner product $-\langle \phi(x),\psi(x)\rangle_{1,\omega}$ given by \eqref{sip}.
\end{proof}
\begin{remark}
Although the weight $\omega(x;q)$ given in \eqref{newweight} seems to satisfy a non-linear relation, it can be explicitly computed if the corresponding weight $\rho(x;q)$ is classical. Examples including the Al-Salam \& Carlitz, $q$-Laguerre, little $q$-Jacobi and big $q$-Jacobi are formulated later.
\end{remark}

\begin{coro}
With $\mathcal{B}_q$ defined in \eqref{relba}, one has
\begin{align}\label{app}
\langle \phi(x),\mb_q\psi(x)\rangle_{2,\rho}=-\langle \psi(x),\mb_q\phi(x)\rangle_{2,\rho}=-\langle \mb_q\phi(x),\psi(x)\rangle_{2,\rho}.
\end{align}
\end{coro}

\subsection{The connection between $q$-SOPs and $q$-OPs}
The construction of the operator $\mb_q$ leads us to finding the connection between $q$-SOPs induced by $\langle\cdot,\cdot\rangle_{1,\omega}$ and $q$-OPs induced by $\langle\cdot,\cdot\rangle_{2,\rho}$.

Let us denote 
\begin{align*}
\mathbf{B}=\left[
\langle p_j,\mb_qp_k\rangle_{2,\rho}
\right]_{j,k=0}^{2N-1},\quad \mathbf{D}=\left[
\langle p_j,p_k\rangle_{2,\rho}
\right]_{j,k=0}^{2N-1},\quad
\mathbf{A}=\left[
\big\langle p_j,\mathcal{A}_qp_k\big\rangle_{2,\rho}
\right]_{j,k=0}^{2N-1},
\end{align*}
so that  one has
\begin{align*}
\left(
\mathbf{B}\mathbf{D}^{-1}\mathbf{A}
\right)_{j,k}&=-(1+\qf )^2\sum_{\ell=0}^{N-1}\frac{\left\langle\mathcal{A}^{-1}_{q}p_j,p_\ell\right\rangle_{2,\rho}}{\big\langle p_\ell,p_\ell\big\rangle_{2,\rho}}\big\langle p_\ell,\mathcal{A}_qp_k\big\rangle_{2,\rho}\\
&=-(1+\qf )^2\left\langle \mathcal{A}_q^{-1}p_j,\sum_{\ell=0}^{N-1}p_\ell\frac{\big\langle p_\ell,\mathcal{A}_qp_k\big\rangle_{2,\rho}}{\langle p_\ell,p_\ell\rangle_{2,\rho}}\right\rangle_{2,\rho},
\end{align*}
where the negative sign comes from the equation \eqref{app}. Applying the equation \eqref{rec} into the formula,
one knows that
\begin{align*}
\mathbf{B}\mathbf{D}^{-1}\mathbf{A}=(1+\qf )^2(\mathbf{D}+\mathbf{C}):=\tilde{\mathbf{D}},
\end{align*}
where $\mathbf{C}$ has all entries zero except the final column.
Denoting $\mathbf{X}$, which is a lower triangular matrix with diagonals $1$, as the connection matrix of $q$-SOPs and $q$-OPs, 
we  have
\begin{align*}
[Q_j(x;q)]_{j=0}^{2N-1}=\mathbf{X}[p_j(x;\qf)]_{j=0}^{2N-1}.
\end{align*}
The skew orthogonality relation \eqref{sor} can be written in the matrix form
\begin{align*}
[\langle Q_j,Q_k\rangle_{1,\omega}]_{j,k=0}^{2N-1}=\mathcal{U},
\end{align*}
where $\mathcal{U}$ is the particular $2\times 2$ block diagonal matrix
\begin{align*}
\mathcal{U}=\text{diag}\left(
\left(
\begin{array}{cc}
0&u_0\\
-u_0&0
\end{array}
\right),\left(
\begin{array}{cc}
0&u_1\\
-u_1&0
\end{array}
\right),\cdots,\left(
\begin{array}{cc}
0&u_{N-1}\\
-u_{N-1}&0
\end{array}
\right)
\right).
\end{align*} 
It follows
\begin{align*}
\mathcal{U}=\left[
\langle Q_j,Q_k\rangle_{1,\omega}
\right]_{j,k=0}^{2N-1}=\mathbf{X}\left[
\langle p_j,p_k\rangle_{1,\omega}
\right]_{j,k=0}^{2N-1}\mathbf{X}^\top=-\mathbf{X}\mathbf{B}\mathbf{X}^\top=-\mathbf{X}\tilde{\mathbf{D}}\mathbf{A}^{-1}\mathbf{D}\mathbf{X}^\top.
\end{align*}
We now intend to solve for the connection matrix $\mathbf{X}$ via the algebraic equation 
$$\mathcal{U}^{-1}\mathbf{X}=-\mathbf{X}^{-\top}\mathbf{D}^{-1}\mathbf{A}\tilde{\mathbf{D}}^{-1}.$$
Noting that $\mathbf{X}$ is a lower triangular matrix with diagonals $1$, one knows that the left hand side of the algebraic equation is an upper Hessenberg matrix and the right hand side is a lower Hessenberg matrix, and thus we can just compute the lower triangular part to obtain coefficients of the connection matrix. By basic matrix operations, one can show that
\begin{align*}
[\mathcal{U}^{-1}\mathbf{X}]_-=\left(
\begin{array}{ccccc}
\ast&&&&\\
1/u_0&\ast&&&\\
-x_{30}/u_1&-x_{31}/u_1&\ast&&\\
x_{20}/u_1&x_{21}/u_1&1/u_1&\ast&\\
\vdots&\vdots&\vdots&\vdots&\ddots
\end{array}
\right).
\end{align*}
Moreover, according to \eqref{rec1}, we have
\begin{align*}
\mathbf{A}=\left(\begin{array}{ccccc}
0&c_0&&&\\
-c_0&0&c_1&&\\
&-c_1&0&c_2&\\
&&\ddots&\ddots&\ddots\end{array}
\right),
\end{align*}
and thus it follows
\begin{align}\label{cj}
-[\mathbf{X}^{-\top}\mathbf{D}^{-1}\mathbf{A}\tilde{\mathbf{D}}^{-1}]_-=\left(\begin{array}{cccc}
\ast&&&\\
\gamma_0&\ast&&\\
&\gamma_1&\ast&\\
&&\ddots&\ddots
\end{array}
\right),\quad \gamma_j=\frac{c_j}{(1+\qf )^2h_jh_{j+1}},
\end{align}
where $c_j=-\langle p_{j+1},\mathcal{A}_qp_j\rangle_{2,\rho}$ and $h_j=\langle p_j,p_j\rangle_{2,\rho}$ are computable if the weight $\rho(x)$ is given.
By equating the above matrices, one can conclude that
\begin{align*}
&\gamma_{2p}=1/u_p\,  (p=0,\cdots,N-1),\quad\quad x_{2p,\ell}=0 \, (\ell=0,\cdots,2p-1),\\
&x_{2p+1,\ell}=0\,(\ell=0,\cdots,2p-2),\quad\,\,\,\, x_{2p+1,2p-1}=-\gamma_{2p-1}u_p=-\gamma_{2p-1}/\gamma_{2p}.
\end{align*}
Hence, the connection between $q$-SOPs $\{Q_j(x;q)\}_{j\in\mathbb{N}}$ when $\beta=1$ and $q$-OPs $\{p_j(x;\qf)\}_{j\in\mathbb{N}}$ can be formulated as
\begin{align}\label{sop-op}
Q_{2j}(x;q)=p_{2j}(x;\qf),\quad Q_{2j+1}(x;q)=p_{2j+1}(x;\qf)+x_{2j+1,2j}p_{2j}(x;\qf)-\frac{\gamma_{2j-1}}{\gamma_{2j}}p_{2j-1}(x;\qf).
\end{align}
Since the skew orthogonality is not changed by the transformation $Q_{2n+1}(x;q)\mapsto Q_{2n+1}(x;q)+\alpha_nQ_{2n}(x;q)$ with constant $\alpha_n$, we set $x_{2j+1,2j}=0$ for simplicity. In this case, the normalisation factor $u_j$ can be simply written as $\gamma_{2j}^{-1}$,  thus providing a simpler way to compute the partition function,
\begin{equation}\label{deb1}
\int_{z_1=0}^{1}\cdots\int_{z_{2m}=0}^{\qf z_{2m-1}} \prod_{1\leq i<j\leq 2m}(z_i-z_j)\prod_{i=1}^{2m}\omega(z_i;q)d_qz_i = \prod_{l=0}^{m-1} u_l = \prod_{l=0}^{m-1}\gamma_{2l}^{-1}.
\end{equation}
We remark too that  $ p_n(x;\qf) $ can be expressed in terms of $ \{Q_n(x;q)\}_{n\geq 0} $. Thus, by solving the relation \eqref{sop-op} backwards, we have
\begin{align}\label{op-sop}
	p_{2j}(x;\qf)=Q_{2j}(x;q),\quad p_{2j+1}(x;\qf)=\sum_{l=0}^j \frac{\prod_{k=1}^ja_k}{\prod_{k=1}^la_k}Q_{2l+1}(x;q),\quad a_k:=\frac{\gamma_{2k-1}}{\gamma_{2k}}.
\end{align}

\section{Examples of $q$-weights and generalised $q$-Selberg integrals}\label{sec4}

\subsection{Little $q$-Jacobi case}
This example  relates directly to  Aomoto's $q$-Selberg integral \eqref{q1}.
The weight function for the little $q$-Jacobi polynomials is
\begin{align*}
\rho(x;q)=x^\alpha(qx;q)_\beta ,\quad \alpha>-1,\,\beta>-1,
\end{align*}
which implies for the Pearson pair
\begin{align*}
(f(x;q),g(x;q))=\left(x(1-x),-([\alpha+\beta+2]_qx-[\alpha+1]_q)\right).
\end{align*}
The construction of the skew-orthogonal weight \eqref{newweight} then demonstrates that 
\begin{align}\label{weight:lqj}
\omega(x;q)=q^{-\frac{\al+1}{8}}x^{\frac{\al-1}{2}}(qx;q)_{\frac{\be-1}{2}}.
\end{align}
One can next apply the Pearson pair to the  computation of the coefficients $c_j$ in equation \eqref{cj}. By taking monic little $q$-Jacobi polynomials $ \{p_n^{(\al,\be)}(x;\qf )\}_{n=0}^{\infty} $ satisfying \eqref{lqj}, one obtains
\begin{align*}
\mathcal{A}_qp_j^{(\al,\be)}(x;\qf )=-q^{-\frac{j}{2}}[\al+\be+2j+2]_\qf  p_{j+1}^{(\al,\be)}(x;\qf )+ \cdots,
\end{align*}
where the terms not shown are lower degree polynomials in the set.
Thus one can see that
\begin{align*}
c_j=-\langle p_{j+1}^{(\al,\be)},\mathcal{A}_\qf  p_j^{(\al,\be)}\rangle_{2,\rho}=q^{-\frac{j}{2}}[\al+\be+2j+2]_\qf  h_{j+1}(\qf ),
\end{align*}
where $h_{j+1}$ denotes the normalisation of $ p_{j+1}^{(\al,\be)}$ as defined in (\ref{3.0}).
Therefore, $\gamma_j$ in \eqref{cj} is exactly
\begin{align*}
\gamma_j=\frac{q^{-\frac{j}{2}}[\al+\be+2j+2]_\qf }{(1+\qf )^2h_j(\qf )}=\frac{q^{-\frac{j^2+\al j+j}{2}}(\qf ;\qf )_{2j+\al+\be+2}(\qf ;\qf )_{2j+\al+\be}}{(1-q)^2(\qf ;\qf )_{j+\al+\be}(\qf ;\qf )_{j+\al}(\qf ;\qf )_{j+\be}(\qf ;\qf )_j}.
\end{align*}

We can combine these explicit values with the general formulas of the previous section to specify the corresponding skew orthogonal polynomials.

\begin{proposition}\label{P4.1}
The little $q$-Jacobi skew orthogonal polynomials $\{Q_n^{(\al,\be)}(x;q)\}_{n=0}^\infty$, which are skew orthogonal under skew inner product defined by \eqref{sip} with weight function \eqref{weight:lqj}, are connected with the little $q$-Jacobi polynomials by the formula
\begin{align*}
Q^{(\al,\be)}_{2j}(x;q )&=p^{(\al,\be)}_{2j}(x;\qf ),\\
Q^{(\al,\be)}_{2j+1}(x;q )&=p^{(\al,\be)}_{2j+1}(x;\qf )\\
&-\frac{q^{-\frac{4j+\al+3}{2}}(1-q^{j+\frac{\al+\be+2}{2}})(1-q^{j+\frac{\al}{2}})(1-q^{j+\frac{\be}{2}})(1-q^j)}{(1-q^{2j+\frac{\al+\be+2}{2}})(1-q^{2j+\frac{\al+\be+1}{2}})(1-q^{2j+\frac{\al+\be}{2}})(1-q^{2j+\frac{\al+\be-1}{2}})}p^{(\al,\be)}_{2j-1}(x;\qf ).
\end{align*}
or equivalently
\begin{align*}
	p_{2j}^{(\al,\beta)}(x;\qf )&=Q_{2j}^{(\al,\beta)}(x;q ),\\
	p_{2j+1}^{(\al,\beta)}(x;\qf )&=\sum_{i=0}^{j}\frac{(q^{i+2+\frac{\al+\beta}{2}},q^{i+1+\frac{\al}{2}},q^{i+1+\frac{\beta}{2}},q^{i+1};q)_{j-i}}{q^{\frac{1}{2}(4i+4j+7+\al)(j-i)}(q^{2i+\frac{\al+\beta+3}{2}};\qf )_{4(j-i)}}Q_{2i+1}^{(\al,\beta)}(x;q ).
\end{align*}
\end{proposition}

In the limit $q \to 1$, upon appropriate normalisation the little $q$-Jacobi weight limits to the Jacobi weight $x^\alpha (1 - x)^\beta$, supported on $(0,1)$.
One can compute the $q\to 1$ limit of Proposition \ref{P4.1}
using the relation \cite[eq. 12.1.7]{ismail05}
\begin{align*}
\lim_{q\to 1}\frac{(q^{\frac{a}{2}};\qf )_j}{(1-\qf )^j}=(a)_j.
\end{align*}
Agreement is found with the known $\beta = 1$ skew orthogonal polynomials for the Jacobi ensemble \cite[Eq.~(3.14)]{adler00}, although
one must be aware that in this latter reference the Jacobi weight is defined on  $(-1,1)$ rather than $(0,1)$. 
Moreover, from our knowledge of $\{ \gamma_{2j}\}$ and (\ref{deb1}) one can get that
\begin{align*}
\int_{z_1=0}^1&\cdots\int_{z_{2n}=0}^{\qf  z_{2n-1}}\prod_{1\leq j<k\leq 2n} (z_k-z_j) \prod_{i=1}^{2n} z_i^{\frac{\al-1}{2}}(qz_i;q)_{\frac{\be-1}{2}} d_qz_i\\
&=q^{\frac{1}{6}n(n-1)(4n+3\al+1)-\frac{1}{4}n(\al+1)}(1-q)^{2n}\prod_{i=0}^{n-1}\frac{(\qf ;\qf )_{2i+\al+\be}(\qf ;\qf )_{2i+\al}(\qf ;\qf )_{2i+\be}(\qf ;\qf )_{2i}}{(\qf ;\qf )_{4i+\al+\be+2}(\qf ;\qf )_{4i+\al+\be}}.
\end{align*}
This result corresponds to the Aomoto's $q$-Selberg integral \eqref{q1} by  changing notation $\alpha\mapsto (\al+1)/2$ and $\be\mapsto(\be+1)/2$ and simple manipulation.

\subsection{Al-Salam $\&$ Carlitz case}
The monic Al-Salam $\&$ Carlitz polynomials $\{U_m^{(\alpha)}(x;q)\}_{m\in\mathbb{N}}$ with weight function
\begin{align*}
\rho(x;q)=(qx;q)_\infty(qx/\alpha;q)_\infty,\quad \alpha<0
\end{align*}
are considered in this part. With this specific weight, the orthogonality relation reads
\begin{align*}
&\int_{\al}^1 U_m^{(\al)}(x;\qf )U_n^{(\al)}(x;\qf )\rho(x;\qf )d_\qf  x\\
&\quad=(-\al)^n(1-\qf ) q^{\frac{n(n-1)}{4}}(\qf ;\qf )_n (\al,\al^{-1}\qf ,\qf ;\qf )_\infty\delta_{n,m}:=h_n(\qf )\delta_{n,m}.
\end{align*}
Moreover, the Pearson pair for the Al-Salam \& Carlitz polynomials is
\begin{align*}
\left(f(x;q ),g(x;q )\right)=\left((1-x)(x-\alpha),\frac{x-(1+\alpha)}{q -1}\right),
\end{align*}
and the corresponding skew-orthogonal weight is
\begin{align}\label{weight:ac}
\omega(x;q)=(-\alpha)^{-1/2}(qx;q)_\infty (qx/\alpha;q)_\infty.
\end{align}
Explicit values of $c_j$ in \eqref{cj} are then given by
\begin{align*}
c_j={q^{-\frac{j}{2}}}{(1-\qf )^{-1}}h_{j+1}(\qf ),
\end{align*}
and thus
\begin{align*}
\gamma_j=\frac{c_j}{(1+\qf )^2h_jh_{j+1}}=-\frac{(-\al)^{-j}q^{-\frac{j(j+1)}{4}}}{(1-q)^2(\qf ;\qf )_j  (\qf ,\al,\al^{-1}\qf ;\qf )_\infty}.
\end{align*}

 The $q$-skew orthogonal system can now be specified.
 
 \begin{proposition}
The Al-Salam $\&$ Carlitz skew orthogonal polynomials $\{Q_n^{(\al)}(x;q)\}_{n=0}^\infty$, which are skew orthogonal under skew inner product defined by \eqref{sip} with weight function \eqref{weight:ac}, are connected with the orthogonal ones $\{U_n^{(\al)}(x;\qf )\}_{n=0}^\infty$ by the formula
\begin{align*}
Q_{2n}^{(\al)}(x;q )=U_{2n}^{(\al)}(x;\qf ),\quad Q_{2n+1}^{(\al)}(x;q )=U_{2n+1}^{(\al)}(x;\qf )+\al q^n(1-q^n)U_{2n-1}^{(\al)}(x;\qf ),
\end{align*}
or equivalently
\begin{align*}
	U_{2n}^{(\al)}(x;\qf )=Q^{(\al)}_{2n}(x;q ),\quad U_{2n+1}^{(\al)}(x;\qf )=\sum_{j=0}^{n}q^{\frac{1}{2}(j+1+n)(n-j)}(-\al)^{n-j}(q^{j+1};q)_{n-j}Q_{2n+1}^{(\al)}(x;q ).
\end{align*}
Moreover, application of (\ref{deb1}) gives
\begin{align*}
\int_{z_1=\al}^1&\cdots\int_{z_{2n}=\al}^{\qf  z_{2n-1}}\prod_{1\leq j<k\leq 2n} (z_k-z_j) \prod_{i=1}^{2n} (qz_i,qz_i/\al;q)_\infty d_qz_i\\
&=(-\al)^{n(n-1)}q^{\frac{1}{12}n(n-1)(4n+1)}(1-q)^{2n}(\qf ,\al,\al^{-1}\qf ;\qf )_\infty^{n}\prod_{i=0}^{n-1}(\qf ;\qf )_{2i}.
\end{align*}
 
\end{proposition}

\subsection{$q$-Laguerre case}
Consider the monic $q$-Laguerre polynomials $ \{p_n^{(\al)}(x;q)\}_{n\in\mathbb{N}} $ with the weight
\begin{align*}
\rho(x;q)={x^{\alpha}}{(-x ; q)^{-1}_{\infty}},\quad \alpha > -1.
\end{align*}
In this case, the orthogonality reads
\begin{multline*}
	\int_{0}^{\infty} {x^{\alpha}}{(-x ; q)^{-1}_{\infty}} p_{m}^{(\alpha)}(x ; q) p_{n}^{(\alpha)}(x ; q) d_{q} x \\
	=\frac{1-q}{2}q^{-2n(n+\al)-n} \frac{\left(q,-q^{\alpha+1},-q^{-\alpha} ; q\right)_{\infty}}{\left(q^{\alpha+1},-q,-q ; q\right)_{\infty}}\left(q^{\alpha+1} ,q; q\right)_{n} \delta_{m n}, \quad \alpha>-1,
\end{multline*}
which gives for the Pearson pair 
\begin{align*}
	\left(f(x;q ),g(x;q )\right)=\left(x,\frac{1}{1-q}(1-q^{\al+1}-q^{\al+1}x)\right).
\end{align*}
Thus, from the formula \eqref{newweight}, the  corresponding skew orthogonal weight function is 
\begin{align}\label{q-lag}
	\omega(x;q)={q^{-\frac{\al+1}{8}}x^{\frac{\al-1}{2}}}{(-x;q)^{-1}_\infty}.
\end{align}
The  computations following \eqref{cj} and \eqref{sop-op} then give
\begin{align*}
	c_n=\frac{q^{\frac{n+\al+1}{2}}}{1-\qf }h_{n+1}(\qf ),\quad
	\gamma_j=\frac{c_j}{(1+\qf )^2h_jh_{j+1}}=\frac{2q^{j^2+\al j+\frac{\al+1}{2}}(q^{\frac{\al+1}{2}},-\qf ,-\qf ;\qf )_\infty}{(1-q)^2(q^{\frac{\al+1}{2}},\qf ;\qf )_j(\qf ,-q^{\frac{\al+1}{2}},-q^{-\frac{\al}{2}};\qf )_\infty},
\end{align*}
allowing  the skew-orthogonal system to be specified.

\begin{proposition}\label{prop_La}
	The $q$-Laguerre skew polynomials $ \{Q_n^{(\al)}(x;q )\}_{n\in\mathbb{N}} $ which are skew orthogonal with respect to the weight function \eqref{q-lag} are connected with the $q$-Laguerre polynomials by the  formula
	\begin{align*}
		Q_{2n}^{(\al)}(x;q )=p_{2n}^{(\al)}(x;\qf ),\quad Q_{2n+1}^{(\al)}(x;q )=p_{2n+1}^{(\al)}(x;\qf )-q^{1-4n-\al}(1-q^n)(1-q^{n+\frac{\al}{2}})p_{2n-1}^{(\al)}(x;\qf ),
	\end{align*}
or equivalently
\begin{align*}
	p_{2n}^{(\al)}(x;\qf )=Q_{2n}^{(\al)}(x;q ),\quad p_{2n+1}^{(\al)}(x;\qf )=\sum_{j=0}^{n}q^{-(2j+2n+1+\al)(n-j)}(q^{j+1},q^{j+1+\frac{\al}{2}};q)_{n-j}Q_{2n+1}^{(\al)}(x;q ).
\end{align*}
	Moreover, we have
	\begin{align*}
		\int_{z_1=0}^\infty&\cdots\int_{z_{2n}=0}^{\qf  z_{2n-1}}\prod_{1\leq j<k\leq 2n} (z_k-z_j) \prod_{i=1}^{2n} \frac{z_i^{\frac{\al-1}{2}}}{(-z_i;q)_\infty} d_qz_i\\
		&=2^{-n}(1-q)^{2n}q^{-\frac{1}{3}n(n-1)(4n-2+3\al)-\frac{3(\al+1)}{4}n}\frac{(\qf ,-q^{\frac{\al+1}{2}},-q^{-\frac{\al}{2}};\qf )_\infty^{n}}{(q^{\frac{\al+1}{2}},-\qf ,-\qf ;\qf )_\infty^n}\prod_{i=0}^{n-1}(q^{\frac{\al}{2}},\qf ;\qf )_{2i}.
	\end{align*}
\end{proposition}
\begin{remark}
	The monic $q$-Laguerre polynomials $ p_{k}^{(\al)}(x;\qf) $ reduces  to the monic Laguerre polynomials $ \tilde{L}_{k}^{(\al)}(x) $ by making transformation $ x\to (1-\qf)x $ and taking the limit $ q\to 1 $ \cite{koekoek96},  i.e.  
	\begin{align*}
		\lim_{q\to 1}\frac{p_{k}^{(\al)}((1-\qf)x;\qf)}{(1-\qf)^k}=   \tilde{L}_{k}^{(\al)}(x).
	\end{align*}
Denote
\begin{align*}
	q_{k}^{(\al)}(x):=\lim_{q\to 1}\frac{Q_{k}^{(\al)}((1-\qf)x;q)}{(1-\qf)^{k}},
\end{align*}
then with the above limiting relation, the first relation in Proposition \ref{prop_La} reduces to 
\begin{align*}
	q_{2n}^{(\al)}(x)=   \tilde{L}_{2n}^{(\al)}(x),\qquad q_{2n+1}^{(\al)}(x)=  \tilde{L}_{2n+1}^{(\al)}(x)-2n(2n+\al) \tilde{L}_{2n-1}^{(\al)}(x),
\end{align*}
which is in agreement with \cite[eq.(3.8)]{adler00}.
\end{remark}

\subsection{Big $q$-Jacobi case}
We turn our attention now to the big $q$-Jacobi polynomials $ J_{n}(x ; a, b, c ; q) $ with weight function
\begin{align*}
	\rho^{(a, b, c)}(x ; q)=\frac{\left(a^{-1} x, c^{-1} x ; q\right)_{\infty}}{\left(x, b c^{-1} x ; q\right)_{\infty}},
\end{align*}
for $ 0<aq<1, 0\leq bq<1 $ and $ c<0 $.
The orthogonality relation of the big $q$-Jacobi polynomials reads
\begin{align*}
		&\int_{c q}^{a q} J_{m}(x ; a, b, c ; q) J_{n}(x ; a, b, c ; q) \rho^{(a, b, c)}(x ; q) d_{q} x =h_{n}(q) \delta_{m n},
\end{align*}
where 
\begin{align*}
	h_{n}(q)=\frac{a q(1-q)\left(-a c q^{2}\right)^{n} q^{\binom{n}{2}}\left(q, a^{-1} c, a c^{-1} q, a b q ; q\right)_{\infty}\left(q, a q, b q, c q, a b c^{-1} q ; q\right)_{n}}{\left(1-a b q^{2 n+1}\right)\left(a q, b q, c q, a b c^{-1} q ; q\right)_{\infty}(a b q ; q)_{n}\left(a b q^{n+1} ; q\right)_{n}^{2}}.
\end{align*}
By solving the Pearson equation \eqref{pearson}, we have
\begin{align*}
	(f(x;q), g(x;q))=\left(\left(1-\frac{x}{a q}\right)\left(1-\frac{x}{c q}\right), \frac{1}{1-q}\left(\left(\frac{1}{a c q^{2}}-\frac{b}{c}\right) x+\frac{b}{c}+1-\frac{1}{a q}-\frac{1}{c q}\right)\right).
\end{align*}
The construction of the weight function for skew orthogonal polynomials can now be carried out, giving
\begin{align}\label{weight:bqj}
	w^{(a, b, c)}(x ; q)=\frac{\left(a^{-1}q^{\frac{1}{2}} x, c^{-1}q^{\frac{1}{2}} x ; q\right)_{\infty}}{\left(x, b c^{-1} x ; q\right)_{\infty}}=\rho^{(aq^{-\frac{1}{2}}, bq^{-\frac{1}{2}}, cq^{-\frac{1}{2}})}(x ; q).
\end{align}
The elements $ c_j $ are then computed as
\begin{align*}
	c_{n}=\frac{a b q^{n+1}-1}{a c(1-q^{\frac{1}{2}}) q^{\frac{n}{2}+1}} h_{n+1}(q^{\frac{1}{2}}),
\end{align*}
and thus
\begin{align*}
	\gamma_j=\frac{(-1)^{j+1}q^{-\frac{1}{4}j^2-\frac{5}{4}j-\frac{3}{2}}(aq^{\frac{1}{2}},bq^{\frac{1}{2}},cq^{\frac{1}{2}},abc^{-1}q^{\frac{1}{2}};q^{\frac{1}{2}})_\infty(abq^{\frac{1}{2}};q^{\frac{1}{2}})_{2j+2}(abq^{\frac{j+1}{2}};q^{\frac{1}{2}})_j}{a^{j+2}c^{j+1}(1-q)^2(q^{\frac{1}{2}},a^{-1}c,ac^{-1}q^{\frac{1}{2}},abq^{\frac{1}{2}};\qf )_\infty(q^{\frac{1}{2}},aq^{\frac{1}{2}},bq^{\frac{1}{2}},cq^{\frac{1}{2}},abc^{-1}q^{\frac{1}{2}};q^{\frac{1}{2}})_j}.
\end{align*}

From this data the skew orthogonal system can be specified.
\begin{proposition}
The skew big $q$-Jacobi orthogonal polynomials $\{Q_n^{(a,b,c)}(x;q)\}_{n=0}^\infty$, which are skew orthogonal under skew inner product defined by \eqref{sip} with weight function \eqref{weight:bqj},
are connected with the big q-Jacobi polynomials $\{	J_{n}^{(a,b,c)}(x ; \qf )\}_{n=0}^\infty$ by the formula
\begin{align*}
	Q_{2n}^{(a,b,c)}(x;q)&=J_{2n}^{(a,b,c)}(x;q^{\frac{1}{2}}),\\ Q_{2n+1}^{(a,b,c)}(x;q)&=J_{2n+1}^{(a,b,c)}(x;q^{\frac{1}{2}})\\
	&+\frac{acq^{n+1}(q^{2n-\frac{1}{2}},aq^{2n-\frac{1}{2}},bq^{2n-\frac{1}{2}},cq^{2n-\frac{1}{2}},abc^{-1}q^{2n-\frac{1}{2}};q^{\frac{1}{2}})_2(1-abq^n)}{(abq^{2n-\frac{1}{2}};q^{\frac{1}{2}})_4}J_{2n-1}^{(a,b,c)}(x;\qf ).
\end{align*}
Furthermore
\begin{align*}
&\int_{z_1=c\qf}^{a\qf}\cdots\int_{z_{2n}=c\qf}^{\qf  z_{2n-1}}\prod_{1\leq j<k\leq 2n} (z_k-z_j) \prod_{i=1}^{2n} \frac{\left(a^{-1}q^{\frac{1}{2}} z_i, c^{-1}q^{\frac{1}{2}} z_i ; q\right)_{\infty}}{\left(z_i, b c^{-1} z_i ; q\right)_{\infty}} d_qz_i\\
&=\frac{(-1)^n(1-q)^{2n}a^{n^2+n}c^{n^2}(q^{\frac{1}{2}},a^{-1}c,ac^{-1}q^{\frac{1}{2}},abq^{\frac{1}{2}};\qf )_\infty^n}{q^{-\frac{n}{12}(4n+5)(n+1)}(aq^{\frac{1}{2}},bq^{\frac{1}{2}},cq^{\frac{1}{2}},abc^{-1}q^{\frac{1}{2}};q^{\frac{1}{2}})_\infty^n}\prod_{i=0}^{n-1}\frac{(q^{\frac{1}{2}},aq^{\frac{1}{2}},bq^{\frac{1}{2}},cq^{\frac{1}{2}},abc^{-1}q^{\frac{1}{2}};q^{\frac{1}{2}})_{2i}}{(abq^{\frac{1}{2}};q^{\frac{1}{2}})_{4i+2}(abq^{i+\frac{1}{2}};q^{\frac{1}{2}})_{2i}}.
\end{align*}
\end{proposition}

\section{Correlation kernels for $q$-orthogonal ensembles} \label{sec5}

 In this section, attention is paid to the $k$-point correlation function for the $q$-orthogonal ensemble in Prop. \ref{prop4}
 \begin{align}\label{cf}
\rho_{N,k}(z_1,\cdots,z_k):= \frac{1}{(N-k)!\tau_{N}}\int_{\mathcal{C}^{N-k}}\Pf [X]\prod_{1\leq i<j\leq N}(z_i-z_j)\prod_{i=k+1}^{N}\omega(z_i;q)d_{\qf }z_i,
 \end{align}
where
\begin{align*}
	X=\left\{
	\begin{array}{ll}
	[s(z_i,z_j)]_{i,j=1}^{2n},&N=2n,\\
	\left[
	\begin{array}{cc}
		[s(z_i,z_j)] & [F(z_i)]\\
		-[F(z_j)]^\top & 0
	\end{array}
	\right]_{i,j=1}^{2n+1}, &N=2n+1,
	\end{array}
	\right.
\end{align*}
$s(x,y)$ and $ F(x) $ are functions defined by \eqref{ssk} and \eqref{F}, $\tau_{N}$ is the partition function of the corresponding $q$-orthogonal ensemble and $\mathcal{C}$ is the corresponding support contour.
Due to the fact that the correlation function have different formula depending on the parity of $ N $, we treat them separately. For more details about Pfaffian point process, one can refer to \cite{borodin05,bufetov21,rains00}.

\subsection{$ N=2n $ even}
In the earlier work \cite{forrester04}, it was shown that the correlation function admitting the form \eqref{cf} for $ N=2n $ can be written as a Pfaffian point process by assuming that
the single moments
\begin{align*}
\xi_i=\int_\mathcal{C}x^i \omega(x;q)d_{\qf }x
\end{align*}
and bi-moments
\begin{align*}
m_{i,j}=\int_{\mathcal{C}^2}s(x,y)x^iy^j\omega(x;q)\omega(y;q)d_{\qf }xd_{\qf }y
\end{align*}
are well-defined and the corresponding skew-symmetric moment matrix is invertible, i.e. the partition function is nonzero.
With these conditions assumed, we can make the Pfaffian point process explicit by specifying the kernel for the
correlation function.

\begin{proposition}\label{P5.1}
The correlation function $\rho_{2n,k}$ in (\ref{cf}) can  be expressed as a Pfaffian according to
\begin{align}\label{5.1a}
\rho_{2n,k}=\Pf\left[
\tilde{K}_{2n}(z_i,z_j)
\right]_{i,j=1}^k\prod_{l=1}^{k}\omega(z_l;q),
\end{align}
where $\tilde{K}_{2n}(x,y)$ is the skew symmetric $2\times 2$ kernel independent of $n$, 
\begin{align*}
\tilde{K}_{2n}(x,y)=\left(
\begin{array}{cc}
K_{2n}(x,y)&J_{2n}(y,x)\\
-J_{2n}(x,y)&-I_{2n}(x,y)
\end{array}
\right).
\end{align*}
With $ \{Q_{l}(x;q)\}_{l\in \mathbb{N}} $ denoted the monic skew orthogonal polynomials with respect to the skew inner product \eqref{sip} under the weight $\omega(z;q)$, $K_{2n}(x,y)$, $ J_{2n}(x,y) $ and $ I_{2n}(x,y) $ are skew symmetric kernels defined by
\begin{align*}
&K_{2n}(x,y)=\sum_{i=0}^{n-1}\frac{1}{u_i}(Q_{2i}(x;q)Q_{2i+1}(y;q)-Q_{2i}(y;q)Q_{2i+1}(x;q)),\\
&J_{2n}(x,y)=s_xK_{2n}(x,y),\quad I_{2n}(x,y)=s(x,y)-s_xs_yK_{2n}(x,y).
\end{align*}
The operator $s_x$ represents the action
\begin{align}\label{s_x}
s_x f(x)=\int_{\mathcal{C}}s(z,x)f(z)\omega(z;q)d_{\qf }z.
\end{align}
\end{proposition}
\begin{proof}
Let $M$ be the moment matrix $(m_{i,j})_{i,j=0}^{2n-1}$.  It is known from \cite[Thm.\,1.1]{rains00}
that the correlation function admits the form $\rho_{2n,k}=\Pf(\tilde{K}_{2n}(z_i,z_j))_{i,j=1}^k$, where with $s^i(x)=\int_{\mathcal{C}} y^i s(x,y)\omega(y;q)d_{\qf }y$
the kernel $\tilde{K}$ is  expressed in terms of the skew symmetric moments according to
\begin{align*}
\tilde{K}_{2n}(x,y)=\left(\begin{array}{cc}
\sum_{0\leq i,j\leq 2n-1}x^iM_{i,j}^{-\top}y^j&\sum_{0\leq i,j\leq 2n-1}x^i M_{i,j}^{-\top}s^j(y)\\
\sum_{0\leq i,j\leq 2n-1}s^i(x)M_{i,j}^{-\top}y^j&-s(x,y)+\sum_{0\leq i,j\leq 2n-1}s^i(x)M_{i,j}^{-\top}s^j(y)
\end{array}
\right).
\end{align*}
By making use of the skew Borel decomposition \eqref{sbd}, one knows that
\begin{align*}
\sum_{0\leq i,j\leq 2n-1}&x^i M_{i,j}^{-\top} y^j=\chi^\top(x) (S^{-1}JS^{-\top})^{-\top} \chi(y)\\
&=\sum_{i=0}^{n-1}\frac{1}{u_i}(Q_{2i}(x;q)Q_{2i+1}(y;q)-Q_{2i}(y;q)Q_{2i+1}(x;q))=K_{2n}(x,y).
\end{align*}
Moreover, by realizing that  the operators $ s_x $ and $ s_y $ only act on the monomials, which doesn't make an influence on the structure of skew orthogonal polynomials, the proof is complete.
 \end{proof}
 
 \begin{remark}
 We use the terminology
 skew $K$-Christoffel-Darboux kernel for the element $K_{2n}(x,y)$ of $\tilde{K}_{2n}(x,y)$.
 It is a reproducing kernel as it satisfies
 \begin{align*}
\int_{\mathcal{C}^2}K_{2n}(x,y)K_{2n}(w,z)s(y,w)\omega(y;q)\omega(w;q)d_{\qf }yd_{\qf }w=K_{2n}(x,z),
 \end{align*}
 as can be checked by making use of the skew orthogonality. Similarly the element $J_{2n}(x,y)$ of $\tilde{K}_{2n}(x,y)$ is to be
 referred to as the  skew $J$-Christoffel-Darboux kernel. According to (\ref{5.1a}), the latter specified the density according to 
 $\rho_{2n,k}(x) = \omega(x;q) J_{2n}(x,x)$.
 \end{remark}
In the subsequent working, we plan to show that for the classical weights $\rho(x;q)$, the  aforementioned skew $J$-Christoffel-Darboux kernel $J_{2n}(x,y)$ is in fact a rank-one perturbation of the classical Christoffel-Darboux kernel 
\begin{align*}
		S_{2n}(x, y): &=\sum_{l=0}^{2n-1} \frac{p_{l}(x;\qf ) p_{l}(y;\qf )}{\langle p_{l}, p_{l}\rangle_{2,\rho}},
\end{align*}
where $ \{p_l(x;\qf )\}_{l\in\mathbb{N}}$ are orthogonal polynomials. Note also that polynomials $\{p_k(x;\qf)\}$ and $\{Q_k(x;q)\}_{k\in\mathbb{N}}$ are two basis in the Hilbert space and so we have the  transition relation
\begin{align}\label{sop2op}
	p_n(x;q^{\frac{1}{2}})=\sum_{k=0}^n\alpha_{n,k}Q_k(x;q ),
\end{align} 
where the coefficients are given in \eqref{op-sop}. The latter can be used to specify the action of the operator $\mathcal B_q$ on the $q$-skew orthogonal polynomials.
\begin{proposition}
We have
\begin{subequations}
	\begin{align}
	&\mb_qQ_{2i}(x;q )=u_i\sum_{k=2i+1}^\infty\frac{p_k(x;q^{\frac{1}{2}})\alpha_{k,2i+1}}{h_k},\label{mq1}\\
	 &\mb_qQ_{2i+1}(x;q )=-u_i\sum_{k=2i}^\infty\frac{p_k(x;q^{\frac{1}{2}})\alpha_{k,2i}}{h_k},\label{mq2}
	\end{align}
	\end{subequations}
	where $\{u_k\}_{k\in\mathbb{N}}$ and $\{h_k\}_{k\in\mathbb{N}}$ are the normalization factors for the  $q$-skew orthogonal polynomials $\{Q_k(x;q)\}_{k\in\mathbb{N}}$ and $q$-orthogonal polynomials $ \{p_k(x;\qf )\}_{k\in\mathbb{N}}$ respectively, and $\{\alpha_{n,k}\}_{k\leq n}$ are the transition coefficients. 
\end{proposition}
\begin{proof}
To prove this proposition, we first need to introduce a delta function $\delta(x,y;\qf): \mathcal{H}\rightarrow \mathcal{H} $, which is defined by
\begin{align*}
	\delta(x, y ; \qf )=\sum_{n=0}^{\infty} \frac{p_{n}(x ; q^{\frac{1}{2}}) p_{n}(y ; q^{\frac{1}{2}})}{h_{n}(\qf )},
\end{align*}
and it admits the following reproducing property
\begin{align}\label{rp}
	\langle\delta(x, y ; \qf ), \xi(y ; \qf )\rangle_{2,\rho}=\xi(x ; \qf ), \quad \forall  \xi(x,\qf )\in \mathcal{H}.
\end{align}
By the relation \eqref{ipr} and \eqref{rp}, we have
	\begin{align*}
		\mb_qQ_{2i}(x;q )=\langle \delta(x, y ; \qf ), \mb_qQ_{2i}(y;q )\rangle_{2,\rho}=-\langle \delta(x, y ; \qf ), Q_{2i}(y;q )\rangle_{1,w}.
	\end{align*}
	Substituting the expansion \eqref{sop2op} in the definition of $ \delta(x,y,\qf ) $ and using the skew orthogonality \eqref{sor}, we get the desired results for the first equation. The second equation is proved similarly.
\end{proof}
Now we are well prepared to relate $J_{2n}$ to the classical Christoffel-Darboux kernel.
\begin{proposition}\label{even_kernel_comp}
	For the classical weights, the skew $J$-Christoffel-Darboux kernel can be expressed as a rank-one perturbation of the classical Christoffel-Darboux kernel by
		\begin{align}\label{5.6a}
		J_{2n}(x,y)
		=&\frac{\rho(x;\qf)}{\omega(x;q)}S_{2n-1}(x,y)+\gamma_{2n-2}\left(s_xp_{2n-2}(x;\qf )\right)p_{2n-1}(y;\qf ),
	\end{align}
	where $ \gamma_{2n-2} $ is defined by \eqref{cj} and $ s_x $ is the operator defined by \eqref{s_x}.
\end{proposition}

\begin{proof}
We first show that 
\begin{align}\label{bqp}
\mb_q\phi(x)=\frac{\omega(x;q)}{\rho(x;\qf)}s_x\phi(x)
\end{align}
holds true for an arbitrary $q$-summable function $\phi(x)$. Since the configuration space is on the half-integer and integer $q$-lattices, we discuss the equation separately. 
By the definition of $ \mathcal{B}_{q} $, we can rewrite $ \mathcal{B}_{q}\phi(q^{m+1/2}) $ in an integration form,
\begin{align*}
\mathcal{B}_q\phi(q^{m+1/2})&=-(1+\qf )(1-q)\sum_{k=-\infty}^{m}\frac{\omega(q^k;q)\omega(q^{k+1/2};q)\dots\omega(q^{m+1/2};q)}{\omega(q^{k+1/2};q)\omega(q^{k+1};q)\dots\omega(q^{m};q)\rho(q^{m+1/2};\qf)}q^k\phi(q^k)\\
&=\frac{\omega(q^{m+1/2};q)}{\rho(q^{m+1/2};\qf)}\int_{q^m}^\infty s(x,q^{m+1/2})\omega(x;q)\phi(x)d_{\qf }x\\
&=\frac{\omega(q^{m+1/2};q)}{\rho(q^{m+1/2};\qf)}\int_{0}^\infty s(x,q^{m+1/2})\omega(x;q)\phi(x)d_{\qf }x.
\end{align*}
Similarly, we have
\begin{align*}
	\mathcal{B}_q\phi(q^{m})=\frac{\omega(q^{m};q)}{\rho(q^{m};\qf)}\int_{0}^\infty s(x,q^{m})\omega(x;q)\phi(x)d_{\qf }x,
\end{align*}
and thus the equation \eqref{bqp} holds. 
Therefore, by acting $\mb_{q,x}$ on $ K_{2n}(x,y) $, we have
\begin{align*}
	\mb_{q,x}K_{2n}(x,y)=&\sum_{i=0}^{n-1}\frac{1}{u_i}(\mb_{q,x}Q_{2i}(x;q)Q_{2i+1}(y;q)-Q_{2i}(y;q)\mb_{q,x}Q_{2i+1}(x;q)),
\end{align*}
from which we can insert formulae \eqref{mq1} and \eqref{mq2} to obtain
\begin{align*}
\mb_{q,x}K_{2n}(x,y)=&\sum_{i=0}^{n-1}\left(Q_{2i}(y;q )\sum_{k=2i}^\infty\frac{p_k(x;q^{\frac{1}{2}})\alpha_{k,2i}}{h_k}+Q_{2i+1}(y;q )\sum_{k=2i+1}^\infty\frac{p_k(x;q^{\frac{1}{2}})\alpha_{k,2i+1}}{h_k}\right)\\
	=&\sum_{k=0}^{2n-1}\sum_{i=0}^k\frac{p_k(x;q^{\frac{1}{2}})\alpha_{k,i}Q_i(y;q )}{h_k}+\sum_{k=2n}^\infty\sum_{i=0}^{2n-1}\frac{p_k(x;q^{\frac{1}{2}})\alpha_{k,i}Q_i(y;q )}{h_k}\\
	=&S_{2n}(x,y)+\sum_{k=2n}^\infty\sum_{i=0}^{2n-1}\frac{p_k(x;q^{\frac{1}{2}})\alpha_{k,i}Q_i(y;q )}{h_k}.
\end{align*}

To evaluate the sum, by making use of the property of the transition coefficients
\begin{align*}
\alpha_{k,i}=\alpha_{k,2n-1}\alpha_{2n-1,i},\quad k\geq 2n,
\end{align*}
we find that 
\begin{align*}
	\sum_{k=2n}^\infty\sum_{i=0}^{2n-1}\frac{p_k(x;q^{\frac{1}{2}})\alpha_{k,i}Q_i(y;q )}{h_k}=&\sum_{k=2n}^\infty\sum_{i=0}^{2n-1}\frac{p_k(x;q^{\frac{1}{2}})\alpha_{k,2n-1}\alpha_{2n-1,i}Q_i(y;q )}{h_k}\\
	=&\sum_{k=2n}^\infty\frac{p_k(x;q^{\frac{1}{2}})\alpha_{k,2n-1}}{h_k}p_{2n-1}(y;q^{\frac{1}{2}})\\
	=&\frac{\mb_qQ_{2n-2}(x;q )p_{2n-1}(y;q^{\frac{1}{2}})}{u_{n-1}}-\frac{p_{2n-1}(x;q^{\frac{1}{2}})p_{2n-1}(y;q^{\frac{1}{2}})}{h_{2n-1}}.
\end{align*}
In view of the definition of $ S_{2n}(x,y) $, substracting the last term from $ S_{2n}(x,y) $ gives $ S_{2n-1}(x,y) $. The proposition is then proved by noting that $ Q_{2n-2}(x;q )=p_{2n-2}(x;\qf ) $ and making use of \eqref{bqp}.

\end{proof}

We can make (\ref{5.6a}) explicit for specific classical $q$-weights.
\begin{enumerate}
\item{little $q$-Jacobi case.} 
\begin{align*}
	J_{2n}(x,y)=&q^{\frac{\al+1}{8}}x^{\frac{\al+1}{2}}(\qf x;q)_{\frac{\be+1}{2}}S_{2n-1}(x, y)-\al_{n} p_{2n-2}^{(\al,\be)}(y;q^{\frac{1}{2}})s_xp_{2n-1}^{(\al,\be)}(x;q^{\frac{1}{2}}),
\end{align*}
where
\begin{align*}
	\al_{n}=\frac{q^{-2n^2+3n-1-\al (n-1)}(q^{\frac{\al+\be-1}{2}+n} ;\qf )_{n}(\qf ;\qf )_{4n-4+\al+\be}}{(1-q)^{2}(\qf ;\qf )_{2n-2+\al}(\qf ;\qf )_{2n-2+\be}(\qf ;\qf )_{2n-2}}.
\end{align*}

\item{Al Salam $\&$ Carlitz case.} 
\begin{align*}
	J_{2n}(x,y)=&(-\al)^{\frac{1}{2}}(\qf x;q)_{\infty}(\qf x/\al;q)_{\infty}S_{2n-1}(x,y)\\
	&+\frac{(-\al)^{2-2n}q^{\frac{-2n^2+3n-1}{2}}}{(1-q)^2(\qf ;\qf )_{2n-2}  (\qf ,\al,\al^{-1}\qf ;\qf )_\infty}U_{2n-2}^{(\al)}(y;q^{\frac{1}{2}})s_xU_{2n-1}^{(\al)}(x;q^{\frac{1}{2}}).
\end{align*}

\item{$q$-Laguerre case.}
\begin{align*}
	J_{2n}(x,y)=&q^{\frac{\al+1}{8}}x^{\frac{\al+1}{2}}(-\qf x;q)_{\infty}^{-1}S_{2n-1}(x, y)\\
	&+\frac{2q^{4n^2+2(\al-4)n+\frac{9-3\al}{2}}(q^{\frac{\al+1}{2}},-\qf ,-\qf ;\qf )_\infty}{(1-q)^2(q^{\frac{\al+1}{2}},\qf ;\qf )_{2n-2}(\qf ,-q^{\frac{\al+1}{2}},-q^{-\frac{\al}{2}};\qf )_\infty}p_{2n-2}^{(\al)}(y;q^{\frac{1}{2}})s_xp_{2n-1}^{(\al)}(x;q^{\frac{1}{2}}).
\end{align*}
\item{big $q$-Jacobi case.} 
\begin{align*}
	J_{2n}(x,y)=\frac{(a^{-1}x,c^{-1}x;q)_{\infty}}{(\qf x,bc^{-1}\qf x;q)_{\infty}}S_{2n-1}(x, y)+\alpha_n J_{2n-2}^{(a,b,c)}(y;q^{\frac{1}{2}})s_xJ_{2n-1}^{(a,b,c)}(x;q^{\frac{1}{2}}),
\end{align*}
where 
\begin{align*}
\alpha_n=\frac{q^{-n^2-\frac{n}{2}+3}(aq^{\frac{1}{2}},bq^{\frac{1}{2}},cq^{\frac{1}{2}},abc^{-1}q^{\frac{1}{2}};q^{\frac{1}{2}})_\infty(abq^{\frac{1}{2}};q^{\frac{1}{2}})_{4n-2}(abq^{\frac{2n-1}{2}};q^{\frac{1}{2}})_{2n-2}}{a^{2n}c^{2n-1}(1-q)^2(q^{\frac{1}{2}},a^{-1}c,ac^{-1}q^{\frac{1}{2}},abq^{\frac{1}{2}};\qf )_\infty(q^{\frac{1}{2}},aq^{\frac{1}{2}},bq^{\frac{1}{2}},cq^{\frac{1}{2}},abc^{-1}q^{\frac{1}{2}};q^{\frac{1}{2}})_{2n-2}}.
\end{align*}
\end{enumerate}
\subsection{$ N=2n+1 $ odd}
In this part, we deal with the correlation function for the $q$-orthogonal ensemble for an odd number particles. The odd particle case is more tedious than the even one, and references in the continuous classical case can be seen in \cite{sinclair09,forrester08,nagao07}.
For simplicity, we set the integrand interval to be $ (0,\infty) $ as done in Section \ref{sec3}. The results for general interval can  be established in a similar manner.

We begin by establishing the odd particle number analogue of Proposition \ref{P5.1}, for which 
a procedure similar to the method used in \cite[Ch.~6.3.3]{forrester10} in the setting of the j.p.d.f.~(\ref{epdf}) with $\beta = 1$ and $N$ odd suffices.
\begin{proposition}\label{prop5.5}
	Let $ \{Q_j(x;q)\}_{j\in \mathbb{N}} $ and $ \{u_j\}_{j\in \mathbb{N}} $ be the skew orthogonal polynomials and the normalization factors defined in Section \ref{sec3}. Define
	\begin{align}
		&\beta_{2n}:=(1-\qf)\sum_{i=-\infty}^{\infty}Q_{2n}(q^{\frac{i}{2}};q)F(q^{\frac{i}{2}})\omega(q^{\frac{i}{2}};q)q^{\frac{i}{2}}, \label{hatu}\\
		&\hat{Q}_k(x;q):=Q_k(x;q)-\frac{(1-\qf)}{\beta_{2n}}\left(\sum_{i=-\infty}^{\infty}Q_{k}(q^{\frac{i}{2}};q)F(q^{\frac{i}{2}})\omega(q^{\frac{i}{2}};q)q^{\frac{i}{2}}\right)Q_{2n}(x;q),\quad (k=0,\dots ,2n-1),\notag\\
		& \hat{Q}_{2n}(x;q):=Q_{2n}(x;q)  \notag\\
		&\Phi_k(x):=(1-\qf)\sum_{i=-\infty}^{\infty}\hat{Q}_{k}(q^{\frac{i}{2}};q)s(q^{\frac{i}{2}},x)\omega(q^{\frac{i}{2}};q)q^{\frac{i}{2}},\quad (k=0,\dots,2n),\notag
	\end{align}
	where $ s(x,y) $ and $ F(x) $ are functions defined by \eqref{ssk} and \eqref{F} respectively. Specifying the skew symmetric $2\times 2$ matrix kernel by
	\begin{align*}
		\tilde{K}_{2n+1}(x,y)=\left(\begin{array}{cc}
			K_{2n+1}^{odd}(x,y) & J_{2n+1}^{odd}(y,x) \\
			-J_{2n+1}^{odd}(x,y) & -I_{2n+1}^{odd}(x,y)
		\end{array}\right),
	\end{align*}
	with 
	\begin{align}
		K_{2n+1}^{odd}(x,y)&=\sum_{k=0}^{n-1}\frac{1}{{u}_k}\left(\hat{Q}_{2k}(x;q)\hat{Q}_{2k+1}(y;q)-\hat{Q}_{2k+1}(x;q)\hat{Q}_{2k}(y;q)\right),\notag\\
		J_{2n+1}^{odd}(x,y)&=\sum_{k=0}^{n-1}\frac{1}{{u}_k}\left(\Phi_{2k}(x)\hat{Q}_{2k+1}(y;q)-\Phi_{2k+1}(x)\hat{Q}_{2k}(y;q)\right)+\frac{1}{\beta_{2n}}F(x)\hat{Q}_{2n}(y;q),\label{odd_kernel}\\
		I_{2n+1}^{odd}(x,y)&=\sum_{k=0}^{n-1}\frac{1}{{u}_k}\left(\Phi_{2k+1}(x)\Phi_{2k}(y)-\Phi_{2k}(x)\Phi_{2k+1}(y)\right)+s(x,y)\notag\\
		&+\frac{1}{\beta_{2n}}\left(\Phi_{2n}(x)F(y)-F(x)\Phi_{2n}(y)\right),\notag
	\end{align}
 the correlation function $ \rho_{2n+1,k} $ can be written as 
	\begin{align*}
		\rho_{2n+1,k}=\prod_{l=1}^{k}\omega(z_l;q)\Pf\left[
		\tilde{K}_{2n+1}(z_i,z_j)
		\right]_{i,j=1}^k.
	\end{align*}
	
\end{proposition}
To establish this for general $k$, we first establish the special case for $ k=2n+1 $. For this it is convenient to make use a quaternion determinant
formalism, a summary of which is given in the  Appendix for the convenience of the reader.

\begin{proposition}
We have
	\begin{align}\label{5.11}
	\prod_{1\leq i<j\leq N}(z_i-z_j)\Pf [X]\prod_{l=1}^{2n+1}\omega(z_l;q)=\operatorname{qdet}[f^{odd}(z_j,z_k)]_{j,k=1}^{2n+1}\prod_{l=0}^{n}u_l,
	\end{align} 
	where
	\begin{align*}
	f^{odd}(x,y):=\left(\begin{array}{cc}
	J_{2n+1}^{odd}(x,y)\omega(y;q) & I_{2n+1}^{odd}(x,y)\\
K_{2n+1}^{odd}(x,y)	\omega(x;q)\omega(y;q) &J_{2n+1}^{odd}(y,x)  \omega(x;q)
	\end{array}\right)
	\end{align*}
	is viewed as a quaternion defined via its matrix representation.
\end{proposition}
\begin{proof}
	By performing  elementary row operations in the Vandermonde determinant, we have
	\begin{align*}
	\begin{aligned}
	\prod_{j=1}^{2n+1} & \omega(z_j;q)^2 \prod_{1 \leq j<k \leq N}\left(z_{k}-z_{j}\right)^{2} \\
	=& \prod_{j=1}^{2n+1} \omega(z_j;q)^2 \operatorname{det}\left[\hat{Q}_{k-1}\left(x_{j};q\right)\right]_{j, k=1, \ldots, 2n+1} \operatorname{det}\left[\begin{array}{c}
	\left[\begin{array}{c}
	\hat{Q}_{2 j-1}(z_{k};q) \\
	-\hat{Q}_{2 j-2}(z_{k};q)\end{array}\right]_{\substack{j=1,\dots,n\\k=1,\dots,2n+1}} \\
	{\left[\hat{Q}_{2n}\left(z_{k};q\right)\right]_{k=1, \ldots, 2n+1}}
	\end{array}\right]\\
	=& \prod_{j=0}^{n} {u}_{j}^{2} \operatorname{det}\left[\mathbf{K}_{1}^{odd}+\vec{q} \vec{q}^{\top}\right],
	\end{aligned}
	\end{align*}
	where 
	\begin{align*}
	\mathbf{K}_{1}^{odd}=[K^{odd}_{2n+1}(z_j,z_k)\omega(z_j;q)\omega(z_k;q)]_{j,k=1}^{2n+1},\quad 
 \vec{q}=\left[\frac{1}{ \beta_{2n}}\hat{Q}_{2n}(z_j;q)\omega(z_j;q)
\right]_{j=1}^{2n+1}.
 \end{align*}
  Since $ \vec{q} $ is a vector, by elementary row operations, one can prove that 
	\begin{align*}
	\operatorname{det}\left[\mathbf{K}_{1}^{odd}+\vec{q} \vec{q}^{\top}\right]=\det\left[\begin{array}{cc}
	\mathbf{K}_{1}^{odd}& \vec{q}\\
	\vec{q}^{\top}& 1
	\end{array}\right].
	\end{align*}
	Notice that $ \mathbf{K}_{1}^{odd} $ is an antisymmetric matrix of odd dimension which has  determinant zero. Thus the element 1 in the lower
	right-hand corner can be replaced by 0, giving an antisymmetric matrix. Then making use of the classical result that for any antisymmetric matrix $ A $
	\begin{align*}
	(\pf A)^2=\det A ,
	\end{align*}
	we have
	\begin{align*}
	\prod_{j=1}^{2n+1} & \omega(z_j;q) \prod_{1 \leq j<k \leq N}\left(z_{k}-z_{j}\right)= \beta_{2n} \prod_{j=0}^{n-1} {u}_{j}\Pf\left[\begin{array}{cc}
	\mathbf{K}_{1}^{odd}& \vec{q}\\
	\vec{q}^{\top}& 0
	\end{array}\right].
	\end{align*}
	Moreover, with $ \vec{F}=[F(z_j)]_{j=1}^{2n+1} $ and $ \mathbf{S}=[s(z_j,z_k)]_{j,k=1}^{2n+1} $, we have
	\begin{align*}
	\prod_{j=1}^{2n+1} \omega(z_j;q) \prod_{1 \leq j<k \leq N}\left(z_{k}-z_{j}\right)\Pf [X]&=(-1)^{n+1} \beta_{2n} \prod_{j=0}^{n-1} {u}_{j}\Pf\left[\begin{array}{cccc}
	\mathbf{K}_{1}^{odd}& \vec{q} & \quad & \quad\\
	\vec{q}^{\top}& 0 &\quad &\quad \\
	\quad & \quad & -\mathbf{S} & -\vec{F}\\
	\quad &\quad & \vec{F}^{\top} & 0
	\end{array}\right]\\
	&=(-1)^{n+1} \beta_{2n}  \prod_{j=0}^{n-1} {u}_{j}\Pf\left[\begin{array}{cc}
	\mathbf{K}_{1}^{odd}& \vec{q}\vec{F}^{\top}\\
	-\vec{F}\vec{q}^{\top}& -\mathbf{S}
	\end{array}\right],
	\end{align*}
	where the validity of the second equality can be checked by the definition of Pfaffian. Denote 
	\begin{align*}
		\mathbf{J}_{1}^{odd}=[J^{odd}_{2n+1}(z_j,z_k)\omega(z_k;q)]_{j,k=1}^{2n+1},\quad \mathbf{I}_{1}^{odd}=[I_{2n+1}^{odd}(z_j,z_k)]_{j,k=1}^{2n+1}
	\end{align*}
	and define the $ z_j $-dependent transformation matrix $ \alpha=[\alpha_{i,j}]_{i,j=1}^{2n+1} $ such that
	\begin{align*}
	\Phi_{k}(z_j)=-\sum_{i=1}^{2n+1}\al_{j,i}\omega(z_i;q)\hat{Q}_{k}(z_i;q).
	\end{align*}
	Then we have 
	\begin{align*}
	\al\mathbf{K}^{odd}_{1}-\vec{F}\vec{q}^\top=-\mathbf{J}^{odd}_{1}, \quad \al\mathbf{K}^{odd}_{1}\al^\top+\al\vec{q}\vec{F}^\top-\vec{F}\vec{q}^\top\al^\top-\mathbf{S}=-\mathbf{I}^{odd}_{1}.
	\end{align*}
	Therefore, by multiplying the first block row of the Pfaffian by $ \al $ on the left and adding it to the second block row, then multiplying the first block column by $ \al^\top $ on the right and adding it to the second block column, we obtain
	\begin{align*}
	 \prod_{1 \leq j<k \leq N}\left(z_{k}-z_{j}\right)\Pf [X]\prod_{j=1}^{2n+1} \omega(z_j;q)&=(-1)^{n+1} \beta_{2n} \prod_{j=0}^{n-1} {u}_{j}\Pf\left[\begin{array}{cc}
	\mathbf{K}_{1}^{odd}& (\mathbf{J}^{odd}_{1})^\top\\
	-\mathbf{J}^{odd}_{1}& -\mathbf{I}^{odd}_{1}
	\end{array}\right]\\
	&= \beta_{2n} \prod_{j=0}^{n-1} {u}_{j}\operatorname{qdet}[f^{odd}(z_j,z_k)]_{j,k=1}^{2n+1},
	\end{align*}
	where the second equality is obtained by reordering the rows and columns by writing first the $2n+1$ odd rows (and columns) before the $2n+1$ even rows (and columns) and then using \eqref{qdet_pf}. 
\end{proof}
For general $ k $,  we appeal a known integration formula relating to a particular class of quaternion determinants and then make an induction on the size of the determinant. 
\begin{proposition}
	\cite[Thm 5.1.4]{mehta04} Let $ K(x,y) $ be a function with real, complex and quaternion values, such that 
	\begin{align*}
	K(x,y)=\bar{K}(y,x),
	\end{align*}
	where $ K = \bar{K} $ if $ K $ is real, $ \bar{K} $ is the complex conjugate of $ K $ if it is complex, and $ \bar{K} $ is
	the dual of $ K $ if it is quaternion. Assume that 
	\begin{align*}
	\int K(x,y)K(y,z)d\mu(y)=K(x,z)+\lambda K(x,z)-K(x,z)\lambda,
	\end{align*}
	with $ \lambda $ a constant quaternion and $ d\mu $ a positive measure. Let $ [K(x_i,x_j)]_{i,j=0\dots,N} $ denote the $ N\times N $ matrix with its $ (i,j) $ element equal to $ K(x_i,x_j) $. Then
	\begin{align*}
	\int \det[K(x_i,x_j)]_{i,j=0\dots,N}d\mu(x_N)=(c-N+1)\det[K(x_i,x_j)]_{i,j=0,\dots,N-1},
	\end{align*}
	where $ c=\int K(x,x)d\mu(x) $.
\end{proposition}
By the orthogonality relations of $ \{\hat{Q}_{k}(x;q)\}_{k\in \mathbb{N}} $, one can check that
\begin{align}
\int_{0}^{\infty}f^{odd}(x,y)f^{odd}(y,z)d_{\qf}y=\left[\begin{array}{cc}
0&0\\
0&1
\end{array}\right]f^{odd}(x,z)+f^{odd}(x,z)\left[\begin{array}{cc}
1&0\\
0&0
\end{array}\right]\label{fodd_1}
\end{align}
and 
\begin{align}
\int_{0}^{\infty}f^{odd}(x,x)d_{\qf}x=(2n+1)\left[\begin{array}{cc}
1&0\\
0&1
\end{array}\right].\label{fodd_2}
\end{align}
Thus, $ f^{odd}(x,z) $ satisfies the conditions in the above proposition.
\begin{proposition}
	The quaternion determinant in (\ref{5.11}) satisfies the integration formula
	\begin{align*}
	\int_{0}^{\infty}\operatorname{qdet}[f^{odd}(z_j,z_k)]_{j,k=1}^{m}d_{\qf}z_m=(2n+2-m)\operatorname{qdet}[f^{odd}(z_j,z_k)]_{j,k=1}^{m-1}.
	\end{align*}
\end{proposition}
Using the integration formula as an inductive step shows
\begin{align*}
\int_{0}^{\infty}\dots \int_{0}^{\infty}\operatorname{qdet}[f^{odd}(z_j,z_k)]_{j,k=1}^{2n+1}\prod_{i=m}^{2n+1}d_{\qf}z_i=(2n+2-m)!\operatorname{qdet}[f^{odd}(z_j,z_k)]_{j,k=1}^{m-1}.
\end{align*}
Taking $ m=1 $ and substituting (\ref{5.11}) give
\begin{align*}
\int_{0}^{\infty}\dots\int_{0}^{\infty} \Pf [X]\prod_{1 \leq j<k \leq N}\left(z_{k}-z_{j}\right) \prod_{j=1}^{2n+1} \omega(z_j;q)d_{\qf}z_j=(2n+1)! \beta_{2n} \prod_{j=0}^{n-1} {u}_{j}
=(2n+1)!\tau_{2n+1},
\end{align*}
where the second equality is the definition of the partition function.
Thus 
\begin{align}\label{partf}
\tau_{2n+1}= \beta_{2n} \prod_{j=0}^{n-1} {u}_{j}.
\end{align}
 Now recalling the definition \eqref{cf} of the correlation functions $ \rho_{2n+1,k} $, it follows from these formulas that
\begin{align*}
\rho_{2n+1,k}=\operatorname{qdet}[f^{odd}(z_i,z_j)]_{i,j=1}^{k}.
\end{align*}
Therefore, Proposition \ref{prop5.5} is now proved after an application of \eqref{qdet_pf}.

Now, we turn our attention to the kernel function $ J_{2n+1}^{odd}(x,y) $ in the odd case, which admits a connection with the skew $J$-Christoffel-Darboux kernel considered in last subsection.
This then allows for an expression in terms of a rank one perturbation of the classical Christoffel-Darboux kernel given in terms of classical $q$-orthogonal polynomials.

\begin{proposition}\label{P5.9}
For the $q$-classical weights, the skew $J$-Christoffel-Darboux kernel $ J^{odd}_{2n+1}(x,y) $ can be expressed as
	\begin{align*}
		J^{odd}_{2n+1}(x,y)=&J_{2n}(x,y)+\frac{1}{\beta_{2n}}F(x)p_{2n}(y;\qf )\\
		&+\frac{\beta_{2n-2}}{{u}_{n-1}\beta_{2n}}\left(s_xp_{2n}(x;\qf )p_{2n-1}(y;\qf )-s_xp_{2n-1}(x;\qf )p_{2n}(y;\qf )\right)\\
		=&\frac{\rho(x;\qf)}{\omega(x;q)}S_{2n-1}(x,y)+\gamma_{2n-2}\left(s_xp_{2n-2}(x;\qf )\right)p_{2n-1}(y;\qf )+\frac{1}{\beta_{2n}}F(x)p_{2n}(y;\qf )\\
		&+\frac{\beta_{2n-2}}{{u}_{n-1}\beta_{2n}}\left(s_xp_{2n}(x;\qf )p_{2n-1}(y;\qf )-s_xp_{2n-1}(x;\qf )p_{2n}(y;\qf )\right),
	\end{align*}
	where, as is consistent with (\ref{hatu}) in the case $i=2n$, 
	\begin{align}\label{bi}
	\beta_i=\int_{0}^{\infty}F(x)Q_i(x;q )\omega(x;q)d_{\qf }x=(1-\qf)\sum_{j=-\infty}^{\infty}F(q^{\frac{j}{2}})Q_i(q^{\frac{j}{2}};q)\omega(q^{\frac{j}{2}};q)q^{\frac{j}{2}}
	\end{align}
	and $ \gamma_{2n-2} $ is defined by \eqref{cj}.
\end{proposition}
\begin{proof}
From the definition \eqref{odd_kernel} of $ J_{2n+1}^{odd}(x,y) $  we have
\begin{align*}
	J_{2n+1}^{odd}(x,y)=&J_{2n}(x,y)+\frac{Q_{2n}(y;q )}{\beta_{2n}}\sum_{k=0}^{n-1}\frac{1}{{u}_k}(\beta_{2k}s_xQ_{2k+1}(x;q )-\beta_{2k+1}s_xQ_{2k}(x;q ))\\
	&+\frac{1}{\beta_{2n}}F(x)Q_{2n}(y;q)+\frac{s_xQ_{2n}(x;q )}{\beta_{2n}}\sum_{k=0}^{n-1}\frac{1}{{u}_k}(\beta_{2k+1}Q_{2k}(y;q )-\beta_{2k}Q_{2k+1}(y;q )).
\end{align*}

The quantity $ J_{2n}(x,y) $ is evaluated according to Proposition \ref{even_kernel_comp}. To compute the remaining two summations, we first notice that by the definition of $ s(x,y) $ and $ F(x) $, we have 
\begin{align*}
 \lim_{m\to\infty}s(q^{m+1/2},x)=F(x)(1+\qf ).
 \end{align*} 
Thus by  \eqref{s_x}, we have 
\begin{align}\label{27.1}
	\lim_{m\to\infty}\frac{s_xQ_i(q^{m+1/2};q )}{(1+\qf )}=-\lim_{m\to\infty}\int_{q^m}^{\infty}\frac{s(q^{m+1/2},x)}{1+\qf }\omega(x;q)Q_i(x;q )d_{\qf }x=-\beta_i.
\end{align}
Therefore, 
provided that 
\begin{align}\label{27.2}
 \lim_{x\to 0}{\rho(x;\qf )}{\omega(x;q)^{-1}}=0,
 \end{align}
the sums can be computed by taking limits of $ s_xK_{2n}(x,y) $ and
\begin{align*}
	\sum_{k=0}^{n-1}\frac{1}{{u}_k}&(\beta_{2k+1}Q_{2k}(y;q )-\beta_{2k}Q_{2k+1}(y;q ))=-\lim_{m\to\infty}\frac{s_xK_{2n}(q^{m+1/2},y)}{(1+\qf )}\\
	&=-\lim_{m\to\infty}\frac{\rho(q^{m+1/2};\qf )S_{2n-1}(q^{m+1/2},y)}{\omega(q^{m+1/2};q)(1+\qf )}+\frac{1}{{u}_{n-1}}\beta_{2n-2}p_{2n-1}(y;\qf )\\
	&=\frac{1}{{u}_{n-1}}\beta_{2n-2}p_{2n-1}(y;\qf ).
\end{align*}

Similarly, by taking limits of $ s_yK_{2n}(x,y) $ and then acting $ s_x $ on both sides of the equation, we have 
\begin{align*}
	\sum_{k=0}^{n-1}\frac{1}{{u}_k}(\beta_{2k}s_xQ_{2k+1}(x;q )-\beta_{2k+1}s_xQ_{2k}(x;q ))=-\frac{1}{{u}_{n-1}}\beta_{2n-2}s_xp_{2n-1}(x;\qf ).
\end{align*}
The proof is then finished by combining the above results and noting that $ Q_{2n}(x;q)=p_{2n}(x;\qf ) $.
\end{proof}

As a final point we evaluate the $q$-integrals  (\ref{bi}) for even subscript. This is relevant to the formula (\ref{partf}) for the partition function,
and to expression given in Proposition \ref{P5.9} for $J^{odd}_{2n+1}(x,y)$.

\begin{proposition}
Define $\gamma_j$ as in (\ref{cj}). The $q$-integrals (\ref{bi}) for even subscript have the evaluation
$$
\beta_{2l} = \beta_0 \prod_{j=0}^{l-1} {\gamma_{2j} \over \gamma_{2j+1}}.
$$
\end{proposition}

\begin{proof}
We consider the formula (\ref{27.1}) for $\beta_{2l}$. Using  (\ref{mq1}), the fact that
$$
\alpha_{k,2i+1} = {1 \over a_i} \alpha_{k,2i-1}
$$
as follows from (\ref{op-sop}), and  (\ref{bqp}), we can check under the assumption (\ref{27.2}) that
$$
\lim_{m \to \infty} s_x Q_{2i}(x;q) \Big |_{x = q^{2m+1}} =
{u_i \over a_i u_{i-1}} \lim_{m \to \infty} s_x Q_{2i-2}(x;q) \Big |_{x = q^{2m+1}}.
$$
According to (\ref{27.1}) this implies
$$
\beta_{2i} = {u_i  \over a_i u_{i-1}} \beta_{2i-2} = {\gamma_{2i-2} \over \gamma_{2i-1} } \beta_{2i-2},
$$
which upon iteration gives the stated expression.
\end{proof}

\section{Further remarks}\label{sec6}
                                                                                                                                                                                                                                                                                                                                                                                                                                                                                                                                                                                                                                                                                                                                                                                                                                                                                                                                                                                                                                                                                                                                                                                                                                                                                                                                                                                                                                                                                                                                                                                                                                                                                                                                                                                               
In this paper, we mainly consider the skew orthogonal polynomials and correlation kernels related to the discrete $q$-orthogonal ensemble, compared with the previous work for discrete $q$-symplectic ensemble  \cite{forrester20}. The orthogonal ensemble is more difficult than the symplectic ensemble to analyze, especially for the configuration space supported on the exponential lattices. We make use of the duality relation between orthogonal ensemble and symplectic ensemble, by which a rank one perturbation of the correlation kernel for $q$-orthogonal ensemble is obtained. 

There are some interesting problems to be continued. One is to find a proper combinatorial model such as enumeration of boxed plane partitions for the corresponding $q$-orthogonal ensemble with certain classical weights. Classical orthogonal polynomials and skew orthogonal polynomials theories will be applied to analyze the corresponding correlation kernels and limiting behaviors will be obtained. Another is about the model of discrete orthogonal ensemble on the linear lattice, motivated by the discrete Selberg integral. For example, in the work \cite{brent16}, the authors considered a discrete analog of Macdonald-Mehta integral
\begin{align*}
S_{r,n}(\al,\gamma,\delta)=\sum_{k_1,\cdots,k_r=-n}^n |\Delta(k^\alpha)|^{2\gamma}\prod_{i=1}^r |k_i|^\delta
{2n\choose {n+k_i}}
\end{align*} 
and in Section 5, they gave many evaluations when $\gamma=1/2$. It seems that the discrete orthogonal ensemble can be embedded into the symmetric configuration space and it needs further investigation.
 
\section*{Acknowledgement}
P. Forrester is supported by the Australian Research Council Discovery Project grant DP210102887. S. Li is partially supported by the National Natural Science Foundation of China (Grant no. 12175155), and G. Yu is supported by National Natural Science Foundation of China (Grant no. 11871336).

\appendix
\section*{Appendix A}
\renewcommand{\thesection}{A}
\setcounter{equation}{0}
In order to prove Proposition \ref{prop5.5}, we need to perform a calculation involving quaternion determinants. We start with the following definition of quaternions.
\begin{definition}
	Quaternions are an algebra with elements of the form
	\begin{align*}
	a_{0}\mathbf{1}+a_{1} \mathbf{e}_1+a_{2} \mathbf{e}_2+a_{3} \mathbf{e}_3, \quad  \mathbf{e}_1^{2}= \mathbf{e}_2^{2}= \mathbf{e}_3^{2}=-1, \quad  \mathbf{e}_1 \mathbf{e}_2 \mathbf{e}_3=-1,
	\end{align*}
	where $ a_0,\dots ,a_3 $ are scalars taking real or complex values.
\end{definition}
The quaternions can be represented by $ 2\times 2 $ complex matrices where the basis elements $ \mathbf{1}, \mathbf{e}_1,  \mathbf{e}_2,  \mathbf{e}_3$ are realized by Borel matrices
\begin{align*}
\mathbf{1}:=\left[\begin{array}{ll}
1 & 0 \\
0 & 1
\end{array}\right] \quad \mathbf{e}_{1}:=i \sigma_{z}=\left[\begin{array}{cc}
i & 0 \\
0 & -i
\end{array}\right] \quad \mathbf{e}_{2}:=i \sigma_{y}=\left[\begin{array}{cc}
0 & 1 \\
-1 & 0
\end{array}\right] \quad \mathbf{e}_{3}:=i \sigma_{x}=\left[\begin{array}{ll}
0 & i \\
i & 0
\end{array}\right].
\end{align*}
The dual of a quaternion $ q=a_{0}\mathbf{1}+a_{1}  \mathbf{e}_1+a_{2}  \mathbf{e}_2+a_{3}  \mathbf{e}_3 $ is defined by
\begin{align*}
\bar{q}=a_{0}\mathbf{1}-a_{1}  \mathbf{e}_1-a_{2}  \mathbf{e}_2-a_{3}  \mathbf{e}_3,
\end{align*}
and $ |q|^2=q\bar{q}=a_0^2+a_1^2+a_2^2+a_3^2$ gives the module of $ q $. For example, for a quaternion $ q $ with representation
\begin{align*}
\begin{bmatrix}
u&v\\
w&z
\end{bmatrix},
\end{align*}
the dual element $ \bar{q} $ is given by
\begin{align*}
\begin{bmatrix}
z&-v\\
-w&u
\end{bmatrix}.
\end{align*}

For a matrix $ Q=[q_{i,j}]_{i,j=1}^{N} $ with quaternion elements, we define its dual as the matrix $ Q^{D}:=[\bar{q}_{j,i}]_{i,j=1}^{N} $ and we use $ \hat{Q} $ to denote the $ 2N\times 2N $ matrix with complex elements obtained by replacing each element $ q_{i,j} $ of $ Q $ with its $ 2\times 2 $ matrix representation.

Moreover, the quaternion determinant for a self-dual quaternion matrix $ Q $ is defined as follows
\begin{definition}
	The determinant qdet of an $ N \times N $ self-dual quaternion matrix $ Q $ is defined as
	\begin{align*}
	\operatorname{ qdet } Q=\sum_{P \in S_{N}}(-1)^{N-l} \prod_{1}^{l}\left(q_{a b} q_{b c} \cdots q_{d a}\right)^{(0)}.
	\end{align*}
	The superscript $ (0) $ means that we take scalar part of the product which is equavalent to $ \frac{1}{2}Tr $ when we compute with matrix representations of quaternions. P is any permutation of the indices
	$ (1,\dots, N) $ consisting of $ l $ disjoint cycles of the form $ a\rightarrow b \rightarrow c\dots \rightarrow d\rightarrow a $ and $ (-1)^{N-1} $ is equal to the
	parity of $ P $.
\end{definition}
Determinants of self-dual quaternion matrices are in close relation with Pfaffians. If we define $ \mathbb{Z}_{2N} $ as the $ 2N\times 2N $ antisymmetric matrix with $ 2\times 2 $ blocks
\begin{align*}
\left[\begin{array}{cc}
0&1\\
-1&0
\end{array}\right]
\end{align*}
along the main diagonal and zeros elsewhere, i.e.
\begin{align*}
\mathbb{Z}_{2N}=\oplus_{i=1}^N\left[\begin{array}{cc}
0&1\\
-1&0
\end{array}\right]
\end{align*}
then it can be checked that a $ N\times N $ quaternion matrix $ Q $ is self-dual if and only if $ \hat{Q}\mathbb{Z}_{2N}^{-1} $ is antisymmetric. Moreover, we have the following.
\begin{proposition}
	For an $ N\times N $ self-dual quaternion matrix $ Q $
	\begin{align}
	\operatorname{qdet}Q=\Pf(\hat{Q}\mathbb{Z}_{2N}^{-1}).\label{qdet_pf}
	\end{align}
\end{proposition}

\small
\bibliographystyle{abbrv}

\end{document}